\documentclass[a4paper,12pt]{article}

\usepackage{amsmath,amsthm,amssymb,amsfonts}
\usepackage{pdfsync}
\usepackage{setspace,indentfirst}
\usepackage[hidelinks]{hyperref}
\usepackage{color}
\usepackage{graphicx,caption}
\usepackage{a4wide}
\newtheorem{theorem}{Theorem}[section]
\newtheorem{proposition}[theorem]{Proposition}

\newtheorem{corollary}[theorem]{Corollary}

\newtheorem{lemma}[theorem]{Lemma}

\theoremstyle{definition}
\newtheorem{definition}[theorem]{Definition}
\newtheorem{remark}[theorem]{Remark}

\newcommand{\dx}{\mathrm{d}x}

\newcommand{\Imm}{\mathrm{Im\,}}

\newcommand{\supp}{\mathrm{supp\,}}

\begin{document}
\title{All self-adjoint extensions of the magnetic Laplacian in nonsmooth domains and gauge transformations}
 \author{C\'esar R. de Oliveira, Wagner Monteiro \\ \small{Departamento de Matem\'atica, UFSCar, S\~ao Carlos, SP, 13560-970 Brazil}} 
\maketitle

\

\abstract{We  use boundary triples to find a parametrization of all self-adjoint extensions of the magnetic Schr\"odinger operator,  in a quasi-convex domain~$\Omega$ with  compact boundary, and   magnetic potentials with components in $\textrm{W}^{1}_{\infty}(\overline{\Omega})$. This gives also a new characterization of all self-adjoint extensions of the Laplacian in nonregular domains.  Then we discuss gauge transformations for such self-adjoint extensions and generalize a characterization of the gauge equivalence of the Dirichlet magnetic operator for the Dirichlet Laplacian; the relation to the Aharonov-Bohm effect, including irregular solenoids, is also discussed.  In particular, in case of (bounded) quasi-convex domains it is shown that if  some extension  is unitarily equivalent (through the multiplication by a smooth unit function) to a realization with zero magnetic potential, then  the same occurs for all self-adjoint realizations. 
}

\

\

\tableofcontents

\

\

\section{Introduction}\label{capIntro}
Let the 1-form $A=\sum_{j=1}^n A_j\mathrm{d}x_j$ be a {\em magnetic potential}  in a subset of~$\mathbb R^n$~\cite{frenkelBook}.   The Schr\"odinger expression 
\[
H^{A}=(-i\nabla+{A})^{2} = \sum_{j=1}^n \left(-i\frac{\partial}{\partial x_j}+A_j\right )^2
\] (in appropriate units)  is the starting point of the description of the behavior of a quantum nonrelativistic particle in~$\Omega\subset\mathbb R^n$ under the influence of the {\em magnetic field} $B=\mathrm{d}A$; in the two and three dimensional cases, $A$ may be identified with a vector field and  $B=\mathrm {rot}\,  A$ (with just one component, i.e., a function, in two dimensions). Together with the Laplacian $H^0=-\Delta$, which corresponds to a zero magnetic potential, its self-adjoint realizations are among the most prominent operators in Mathematics and Physics.  The self-adjointness is a requirement for the operator to describe an energy observable in Quantum Mechanics.

Usually it is not a trivial step to classify all self-adjoint extensions of a given symmetric differential operator, especially in domains with boundary irregularities; even without mentioning that both deficiency indices are infinite in this situation. One of the main purpose of this work is to describe all self-adjoint extensions of~$H^A$ in a quasi-convex~\cite{ADO} open set $\Omega\subset \mathbb R^{n}$, $n\ge2$, with compact boundary~$\partial\Omega$,~${A}$ a vector field with components in $\textrm{W}^{1}_{\infty}(\overline{\Omega})$, and initial operator domain~$\mathrm{C}^{\infty}_{0}(\Omega)$ (see Theorem~\ref{extaaU}); by now it is enough to mention that the class of quasi-convex domains contains all convex domains and all domains of class~$\mathrm C^{1,r}$, for~$r > 1/2$. Another purpose is to apply this parametrization to study some properties of all such self-adjoint extensions under gauge transformations (Theorems~\ref{gauge} and~\ref{twii}). The first goal is achieved by following  ideas of~\cite{ADO}, where the authors deal with the Laplacian, but here such ideas are supplemented with the construction of  boundary triples. By restricting ourselves to suitable bounded magnetic potentials, we were able to keep the hypothesis of quasi-convexity, and additionally of considering unbounded~$\Omega$ with compact boundaries.

Particularly attractive is the case in which~$\Omega$ is not simply connected and the magnetic field is zero in~$\Omega$, but $A\ne0$. This opens the possibility for the Aharonov-Bohm effect~\cite{aharonovbohm, Tonomura1989,EFAB}, and  our results give the first description of all of its self-adjoint extensions as well, and for irregular solenoids (see~\cite{adamiTeta,dabrowskiStovicek} for the case $\Omega=\mathbb R^2\setminus\{0\}$). In~\cite{EFAB} there is an interesting characterization of the absence of this effect for the Dirichlet extension, that is, when the  Dirichlet extension is  unitarily equivalent to the case with~$A=0$ (i.e., the Dirichlet Laplacian);  by using gauge transformations, we have then a version of this result in all self-adjoint extensions, but for bounded and connected  (although irregular)  regions. By Theorem~\ref{twii}, one concludes that if one extension of~$H^A$ is gauge unitarily equivalent  to an extension of the Laplacian, then  the same occurs for all self-adjoint realizations, and this is independent of the spectral type of each realization. This seems to be the first proof of such physically expected phenomenon, and here for (bounded regions) and irregular quasi-convex domains.

The characterization of all self-adjoint extensions in the magnetic case, in quasi-convex domains, has some differences with respect to the Laplacian case originally discussed in~\cite{ADO}; besides the presence of the magnetic potential~$A$ in many estimates, the main differences are related to the Neumann trace map, which influences an integration by parts formula, and the introduction of the space~$N^{3/2}_{A}(\partial\Omega)$  (Definition~\ref{defN32A}) used in the extension of the modified Neumann trace to the domain of the maximal operator. It is interesting to note that, differently from~\cite{ADO}, we use boundary triples so that our parametrization of all self-adjoint extensions is in terms of unitary operators on the space~$N^{\frac12}(\partial\Omega)$, and the set of such unitary operators is independent of the magnetic potential~$A$ (see Theorem~\ref{extaaU}). We think this is a more transparent construction, particularly we use (bounded) unitary operators on  boundary spaces; moreover, we have a natural bijection among such extensions for different magnetic potentials (see Remark~\ref{remarkBijecU}), and this discussion is not restricted to bounded~$\Omega$ (although its boundary is supposed to be compact). Complementary to~\cite{ADO,EDOL}, for the case~$A=0$ we have got a new parametrization of all self-adjoint extensions of the Laplacian in (possibly unbounded) quasi-convex domains.

 In~\cite{ACON}, there is, in particular, a characterization of all self-adjoint extensions of  minimal symmetric elliptic differential operators of even-order in~$\mathrm{L}^2(\Omega)$,  for smooth~$\Omega$ (see also~\cite{posilicano} for a general parametrization that reduces to results in~\cite{ACON} ).  For an interesting discussion about the motivations for considering quasi-convex domains, the differences from the approach of~\cite{ACON} and its relations to the Laplacian, see the Introduction of~\cite{ADO}, which is our primary reference for the discussion of extensions.  Here we just mention some points. For smooth~$\Omega$, the domains of the  Dirichlet and Neumann Laplacians are subspaces of~$H^2(\Omega)$, whereas for Lipschitz~$\Omega$ the domain of the Dirichlet Laplacian is a subspace of~$H^{\frac32}(\Omega)$, and this can not be improved in general (one needs an additional effort to get similar results for  some quasi-convex domains).  For Lipschitz domains the range of the combined Dirichlet and Neumann traces, defined on~$H^2(\Omega)$, is not a Cartesian product of boundary Sobolev spaces (see Theorem~\ref{gamma2}).  By using the concept of {\it almost boundary triples}, in~\cite{EDOL} the authors have found a parametrization of the family of all self-adjoint extensions of the minimal Laplacian in  Lipschitz domains. Although more general than the quasi-convex domain case, the cost for the larger generality is a more abstract construction, and so more  difficult to work with in applications; for instance, as already mentioned, the domain of the Dirichlet and Neumann Laplace operators  are not contained in~$H^{2}(\Omega)$, and this regularity  is fundamental in some explicit calculations in the quasi-convex  case. 
  
  The concept of quasi-convexity  is a balance that has permitted a characterization of all self-adjoint extensions of the Laplacian, including the magnetic one, which is not too abstract; hence we  restrict ourselves to the quasi-convex case in this work.

  In Chapter~\ref{cap2} we  recall some basics facts about Sobolev spaces in Lipschitz domains and  Dirichlet trace on Sobolev spaces; many notations are introduced. In Chapter~\ref{cap3} we will introduce the maximal and minimal magnetic operators and show how they are related to each other, and then state some integration by parts formulas related to~$H^{A}$, as well as some density results that will be important in the rest of this work. In Chapter~\ref{cap4} we will extend the magnetic Dirichlet and  Neumann trace operators to the domain of the maximal operator; the range of these operators are the important spaces $N^{1/2}(\partial\Omega)$ and $N^{3/2}_{A}(\partial\Omega)$, respectively. We also briefly review the concept of quasi-convex domains and present one regularity result about the domain of the Dirichlet extension; this is a slight generalization of a regularity result obtained in~\cite{ADO}. In Chapter~\ref{cap7} we briefly review the concept of boundary triples and use it to obtain a parametrization of the family of all self-adjoint extensions of~$H^A$ in a quasi-convex domain; then we use this parametrization to set some results about gauge equivalence of self-adjoint extensions corresponding to the two operators $H^{A},H^{B}$, where~$A$ and~$B$ are two  gauge equivalent magnetic potentials and have components in $\textrm{W}^{1}_{\infty}(\overline{\Omega})$. Applications to the Aharonov-Bohm setting also appear in this section.

\

\noindent {\bf Acknowledgments:} CRdO thanks the partial support by CNPq (Brazilian agency, under contract 303503/2018-1), and WM was supported by CAPES (Brazilian agency).

\

\section{Basics of Sobolev spaces on Lipschitz domains}\label{cap2}
 In this section we  recall some basics facts and notations about Sobolev spaces necessary for this work, including the notions of Dirichlet and magnetic Neumann trace operators. For more details, definitions and  proofs see, for example, \cite{STS}.
 
\subsection{Sobolev spaces and Lipschitz domains}
  An open set~$\Omega$ of $\mathbb R^{n}$, with $n\geq 2$, is said to be a Lipschitz domain if there exists an open cover $\{O_{i}\}_{0\leq i\leq k}$, of its boundary $\partial\Omega$, such that for $i\in \{1,...,k\}$, $O_{i}\cap\Omega$ is equal to the part of $O_{i}$  below the graph of a Lipschitz function $\varphi_{i}:\mathbb R^{n-1}\rightarrow \mathbb{R}$ (considered, possibly, in a coordinate system obtained by a rigid motion). In a similar way we can define a domain of class $\mathrm{C}^{1,r}$, the only difference is that the functions~$\varphi_{i}$  are supposed  of class~$\mathrm{C}^{1,r}$.

Given an open set~$\Omega$ of $\mathbb R^{n}$ and  $s\in \mathbb{R}$, we denote by~$H^{s}(\Omega)$ the corresponding Sobolev space~\cite{STS}. For the same open set we can introduce the  space
\begin{equation*}
H^{s}_{0}(\Omega) =\left\{ u\in H^{s}(\mathbb R^{n})|\; \supp(u) \subseteq \overline{\Omega} \right\}, \quad s\in \mathbb{R},
\end{equation*}
which is equipped with the norm induced by $H^{s}(\mathbb R^{n})$. We also introduce the following spaces,
\begin{equation*} 
\dot{H}^{s}(\Omega)= \overline{\mathrm{C}^{\infty}_{0}(\Omega)}\quad \mathrm{in} \quad H^{s}(\Omega)
\end{equation*}
and
\begin{equation*} 
\tilde{H}^{s}_{0}(\Omega)= \left\{u \in H^{s}(\Omega)|\;u=U_{|\Omega}\quad\mathrm{with}\quad U \in H^{s}_{0}(\Omega)\right\}\,.
\end{equation*}
It is possible to prove that, for any open set $ \Omega$, we have\footnote{In this work the notation $X^{*}$ will be always used to denote the adjoint space of~$X$, that is , the space of continuous antilinear functionals of X.}
\begin{equation}\label{isob}
(H^{s}_{0}(\Omega))^{*}=H^{-s}(\Omega)\quad \mathrm{with}\quad s \in \mathbb{R}
\end{equation}
and
\begin{equation}\label{isob2}
(H^{s}(\Omega))^{*}=H^{-s}_{0}(\Omega)\quad \mathrm{with}\quad s \in \mathbb{R}\,.
\end{equation}
 More precisely,  we are identifying an element $u\in H^{-s}(\Omega)$ with the anti-linear functional $f_{u}$ in $H^{s}_{0}(\Omega)$ given by 
\[
 f_{u}(v)=\int_{R^{n}}\overline{v}Ud^{n}x, 
\]
where $u=U|_{\Omega}$ , $U\in H^{-s}(R^{n})$ and $v \in H^{s}_{0}(\Omega)$. A similar identification is made in~\eqref{isob2}.
 
It is also possible to prove that, if $\Omega$ is a Lipschitz domain with compact boundary, we have
\begin{equation*}\label{e1}
\dot{H}^{s}(\Omega)=\tilde{H}^{s}_{0}(\Omega)\quad \mathrm{for}\quad s>-\frac{1}{2},\quad s\in \mathbb{R}\setminus\left\{\frac{1}{2}+N\right\},\quad N\in\mathbb N,
\end{equation*}
and
\begin{equation*}\label{e2}
(H^{s}(\Omega))^{*}=H^{-s}(\Omega)\quad \mathrm{with}\quad -\frac{1}{2}<s<\frac{1}{2}\,.
\end{equation*}

For a Lipschitz domain with compact boundary it is possible to define the following Sobolev spaces over the boundary  $\partial\Omega$,
\begin{equation*}
H^{s}(\partial\Omega)\quad\mathrm{with}\quad -1\leq s\leq 1;
\end{equation*}
again see~\cite{STS} for the definitions, and the following holds,
\begin{equation*}
(H^{s}(\partial\Omega))^{*}=H^{-s}(\partial\Omega)\quad \mathrm{with}\quad -1\leq s\leq 1.
\end{equation*}

 Recall that for a Lipschitz domain~$\Omega$ it is possible to define, in almost all points of the boundary~$\partial\Omega$, a unit normal vector field $\nu=(\nu_{1},...,\nu_{n})$ pointing outward to~$\Omega$. Let, then,~$\Omega$ be a Lipschitz domain with compact boundary; for each function~$g$ of class~$\mathrm{C}^{1}$ in a neighborhood of~$\partial\Omega$ we consider 
\begin{equation*}
\frac{\partial g}{\partial\tau_{j,k}}=\nu_{i}(\partial_{k}g)|_{\partial\Omega}-\nu_{k}(\partial_{i}g)|_{\partial\Omega}.
\end{equation*}
Given a function $f\in \mathrm{L}^{1}(\partial\Omega)$, consider the functional $\frac{\partial f}{\partial\tau_{j,k}}$ given by
\begin{equation*}
\frac{\partial f}{\partial\tau_{j,k}}:g\in \mathrm{C}^{1}(\partial\Omega)\rightarrow \int_{\partial\Omega}\mathrm{d}^{n-1}\omega\;  f\frac{\partial g}{\partial\tau_{k,j}}\,,
\end{equation*}
where $\mathrm{d}^{n-1}\omega$ is the surface measure on~$\partial\Omega$ (it is well defined by Rademacher Theorem). It is possible to show that for $s\in[0,1]$ the operator $\frac{\partial }{\partial\tau_{j,k}}$ maps $H^{s}(\partial\Omega)$ into $H^{s-1}(\partial\Omega)$ continuously. Furthermore, in~\cite{GRB} one finds the proofs of the following lemmas:
\begin{lemma}
Let~$\Omega$ be a Lipschitz domain with compact boundary. Then for $s\in[0,1]$ we have
\begin{equation*}
H^{s}(\partial\Omega)=\left\{f\in \mathrm{L}^{2}(\partial\Omega,\mathrm{d}^{n-1}\omega)\Big|\;\frac{\partial f}{\partial\tau_{j,k}}\in H^{s-1}(\partial\Omega), 1\leq j,k\leq n\right\}
\end{equation*}
and
\begin{equation*}
\parallel f\parallel_{H^{s}(\partial\Omega)} \approx \parallel f\parallel_{\mathrm{L}^{2}(\partial\Omega,\mathrm{d}^{n-1}\omega)}+\sum_{j,k=1}^{n}\parallel\frac{\partial f}{\partial\tau_{j,k}}\parallel_{ H^{s-1}(\partial\Omega)}.
\end{equation*}
\end{lemma}

With those concepts in mind, we can introduce,
\begin{equation*}
\nabla_{\mathrm{\mathrm{tan}}}:H^{1}(\partial\Omega)\rightarrow \mathrm{L}^{2}_{\mathrm{tan}}(\partial\Omega,\mathrm{d}^{n-1}\omega),
\end{equation*}
\begin{equation*}
\nabla_{\mathrm{tan}}=\left(\sum_{k=1}^{n}\nu_{k}\frac{\partial }{\partial\tau_{k,1}},...,\sum_{k=1}^{n}\nu_{k}\frac{\partial }{\partial\tau_{k,n}}\right)
\end{equation*}
where,
\begin{eqnarray*}
\mathrm{L}^{2}_{\mathrm{tan}}(\partial\Omega,\mathrm{d}^{n-1}\omega)&=&\big\{ f=(f_{1},...,f_{n})|\;f_{i}\in \mathrm{L}^{2}(\partial\Omega,\mathrm{d}^{n-1}\omega), \;i=1,...,n,\quad\\
  & & \nu\cdot f=0\quad \omega{\text{-}} \mathrm{a.s.\; in}\;\partial\Omega\big\}. 
\end{eqnarray*}
We will also make use of the following notation, the duality pairing between $H^{s}(\Omega)$ and $(H^{s}(\Omega))^{*}$,
\begin{equation*}
\langle\cdot,\cdot\rangle_{s}:H^{s}(\Omega)\times H^{s}(\Omega)^{*} \rightarrow \mathbb C.
\end{equation*}

Let  $\textrm{W}^{q}_{\infty}(\overline{\Omega})$ denote  the usual Sobolev space with $q\in\mathbb N$, that is, the set of the restrictions $u=G_{|\Omega}$ of $G\in \textrm{W}^{q}_{\infty}(\mathbb{R}^{n})$  equipped with the norm 
\[
\|u\|_{\textrm{W}^{q}_{\infty}(\overline{\Omega})}=\inf_{G_{|\Omega}=u}\|G\|_{\textrm{W}^{q}_{\infty}(\mathbb{R}^{n})}.
\] It is known that  $u\in\textrm{W}^{1}_{\infty}(\mathbb{R}^{n})$ if, and only if, there is~$K>0$, that depends on~$u$, such that $u$ is a bounded locally $K$-Lipschitz function, in particular it follows that, if $u\in \textrm{W}^{1}_{\infty}(\overline{\Omega})$, then~$u$ is the restriction of a bounded locally $K$-Lipschitz function  defined on~$\mathbb{R}^{n}$.

\begin{remark}\label{rem3}
Note that if   $f\in\textrm{W}^{q}_{\infty}(\overline{\Omega})$, then the operator of multiplication by~$f$, $ M_{f}$, maps $H^{q}(\Omega)$ into itself, moreover, $M_{f}|_{H^{q}(\Omega)}:H^{q}(\Omega)\rightarrow H^{q}(\Omega)$ is bounded. This  result can be seen as particular case of Theorem~3.2 of~\cite{SRAC}.
\end{remark}

\subsection{Dirichlet and  magnetic Neumann  trace operators}

Let~$\Omega$ be  a Lipschitz domain with compact boundary, then the  Dirichlet trace operator
\begin{equation*}
\gamma_{\mathrm{D}}^{0}:\mathrm{C}(\overline{\Omega})\rightarrow\partial\Omega,\quad \gamma_{\mathrm{D}}^{0}u=u|_{\partial\Omega}
\end{equation*}
has the following extensions
\begin{eqnarray}\label{trd}
\gamma_{\mathrm{D}}&:&H^{s}(\Omega)\rightarrow H^{s-\frac{1}{2}}(\partial\Omega),\quad \frac{1}{2}<s<\frac{3}{2}\,,
\\\gamma_{\mathrm{D}}&:&H^{\frac{3}{2}}(\Omega)\rightarrow H^{1-\epsilon}(\partial\Omega),\quad \forall \epsilon \in(0,1)\,,
\end{eqnarray}
furthermore $\gamma_{\mathrm{D}}:H^{s}(\Omega)\rightarrow H^{s-\frac{1}{2}}(\partial\Omega),\; \frac{1}{2}<s<\frac{3}{2}$, has a bounded right inverse, in particular  it is onto.
In most of this work we will be concerned about the following operators related to the bounded vector field~${A}$ with components in $\textrm{W}^{1}_{\infty}(\overline{\Omega})$,
\begin{equation} \label{opeHA}
\nabla_{A}:= \nabla+i{A}\quad \mathrm{and} \quad   H^{A}=(-i\nabla+{A})^{2}.
\end{equation}

\begin{remark}\label{rem1}
For $u\in \mathrm{L}^{2}(\Omega)$ we can define $H^{A}u$ as a distribution acting on $\textrm{C}^{\infty}_{0}(\Omega)$ by $\langle H^{A}u,\varphi\rangle=(\overline{u},H^{A}\varphi)_{\mathrm{L}^{2}(\Omega)}$ , for $\varphi\in\textrm{C}^{\infty}_{0}(\Omega)$. To see that this definition makes sense, note that the following holds in the sense of the distributions  
\[
H^{A}u=-\triangle u-2i{A}\cdot\nabla u+(|{A}|^{2}-i\,\mathrm{div}{A})u\,,
\] where the term  ${A}\cdot\nabla u$ is well defined as a distribution; indeed, since  $A_{i}\in \textrm{W}^{1}_{\infty}(\overline{\Omega}),$ $i=1,...,n$ and $\nabla u\in (H^{-1}(\Omega))^{n}$, it follows that ${A}\cdot\nabla u \in H^{-1}(\Omega)$, so it is a distribution. Note that if $u\in \mathrm{L}^{2}(\Omega)$ and $H^{A}u\in \mathrm{L}^{2}(\Omega)$, for $\varphi\in \textrm{C}^{\infty}_{0}(\Omega)$ we have 
\[
(H^{A}u,\varphi)_{\mathrm{L}^{2}(\Omega)}=\overline{\langle H^{A}u,\overline{\varphi}\rangle}=(u,H^{A}\varphi)_{\mathrm{L}^{2}(\Omega)}\,.
\] 
\end{remark}

We  introduce the {\it magnetic Neumann trace operator}
\begin{equation}\label{TN}
\gamma^{A}_{N}:=\nu\cdot\gamma_{\mathrm{D}}\nabla_{A}:H^{s+1}(\Omega)\rightarrow \mathrm{L}^{2}(\partial\Omega),\quad \frac{1}{2}<s<\frac{3}{2}\,;
\end{equation}
of course, if ${A}=0$, then this is the usual Neumann trace operator and simply denoted by~$\gamma_{\mathrm{N}}$.

\section{Minimal and maximal operators; integration by  parts formula} \label{cap3}

In this section we describe the  initial, maximal and minimal operators associated with the formal operator~$H^{A}=(-i\nabla+{A})^{2}$; then we  set some integration by parts identities related to these operators in Lipschitz domains.

\subsection{Initial, minimal and maximal operators}
Let~$\Omega$ be an open set of $\mathbb R^{n}$; we define the initial operator $H^{A}_{0}$ by
\begin{equation*}
\mathrm{Dom}\,(H^{A}_{0})=\mathrm{C}^{\infty}_{0}(\Omega)\quad\mathrm{and}\quad H^{A}_{0}u:=H^{A}u,
\end{equation*}
which is symmetric and densely defined, and so it is closable  and its closure is denoted by $H^{A}_{\mathrm{min}}$, the so-called {\it minimal (magnetic) operator}. The {\it maximal (magnetic) operator}   $H^{A}_{\mathrm{max}}$ is given by 
\begin{equation*}
\mathrm{Dom}\,H^{A}_{\mathrm{max}}=\left\{ u\in \mathrm{L}^{2}(\Omega)|\; H^{A}u\in \mathrm{L}^{2}(\Omega)\right\},\quad H^{A}_{\mathrm{max}}:=H^{A}u.
\end{equation*}

The proof of Lemma~\ref{lemmaBerndtLaplac}, for the Laplace operator, is due to Prof.~J.~Behrndt (private communication); however, the magnetic potentials introduce additional difficulties.

\begin{lemma}\label{lemmaBerndtLaplac}
Let~$\Omega$ be a Lipschitz domain with compact boundary. Then 
\begin{equation*}
(\overline{H^{A}_{0}})^*=(H^{A}_{\mathrm{min}})^{*}=H^{A}_{\mathrm{max}}.
\end{equation*} 
If $\Omega$ is also bounded, then 
\begin{equation*}
\mathrm{Dom}\,H^{A}_{\mathrm{min}}= H^{2}_{0}(\Omega)\,.
\end{equation*} 
\end{lemma}
\begin{proof}
We  first show that $(\overline{H^{A}_{0}})^*=H^{A}_{\mathrm{max}}$ or, equivalently, $(H^{A}_{0})^{*}=H^{A}_{\mathrm{max}}$. Pick $f\in \mathrm{Dom}\,(H^{A}_{0})^{*}$, then $f\in \mathrm{L}^{2}(\Omega)$ and  there exists $g\in \mathrm{L}^{2}(\Omega)$ such that,
\begin{equation*}
(g,u)_{\mathrm{L}^{2}(\Omega)}=(f,H^{A}u)_{\mathrm{L}^{2}(\Omega)},\quad \forall u\in \mathrm{C}^{\infty}_{0}(\Omega)\,.
\end{equation*}
 So,  by Remark~\ref{rem1}, we have $H^{A}f=g\in \mathrm{L}^{2}(\Omega)$ in the sense of distributions; hence $f\in\mathrm{Dom}\,H^{A}_{\mathrm{max}}$ and  $H^{A}f=g=(H^{A}_{0})^{*}f$, and so $(H^{A}_{0})^{*}\subset H^{A}_{\mathrm{max}}$. On the other hand, for $f\in\mathrm{Dom}\,H^{A}_{\mathrm{max}}$, we have
\begin{equation*}
(H^{A}f,u)_{\mathrm{L}^{2}(\Omega)}=(f,H^{A}u)_{\mathrm{L}^{2}(\Omega)},\quad \forall u\in \mathrm{C}^{\infty}_{0}(\Omega),
\end{equation*}
and it follows immediately that  $ H^{A}_{\mathrm{max}}\subset(H^{A}_{0})^{*}$; hence $(\overline{H^{A}_{0}})^*=H^{A}_{\mathrm{max}}$.

Now assume that~$\Omega$ is bounded. Note that the norm 
\begin{equation*}
\|\cdot\|_{A}=\big(\|\cdot\|_{\mathrm{L}^{2}(\Omega)}^{2}+\|H^{A}(\cdot)\|_{\mathrm{L}^{2}(\Omega)}^{2}\big)^{1/2}
\end{equation*}
is equivalent to the norm of $H^{2}(\Omega)$ over $H^{2}_{0}(\Omega)$.
 Indeed, since ${A}\in \textrm{W}^{1}_{\infty}(\overline{\Omega})$, it is  clear that $\|u\|_{A}\leq C\|u\|_{H^{2}(\Omega)}$ for $u\in H^{2}_{0}(\Omega)$ with an appropriate $C>0$. On the other hand, since $H^{A}=-\triangle+L$, where~$L$ is the linear operator of order~1 given by $Lu=-2i{A}\cdot\nabla u+(|{A}|^{2}-i\,\mathrm{div}{{A}})u$, it follows that, for $u\in H^{2}_{0}(\Omega)$,
\begin{eqnarray}\label{w1}
\|\triangle(u)\|_{\mathrm{L}^{2}(\Omega)}&\leq &\|H^{A}(u)\|_{\mathrm{L}^{2}(\Omega)}+\|Lu\|_{\mathrm{L}^{2}(\Omega)}\nonumber\\
 &=&\|H^{A}(u)\|_{\mathrm{L}^{2}(\Omega)}+C_{0}\|u\|_{\mathrm{L}^{2}(\Omega)}+C_{1}\|\nabla u\|_{\mathrm{L}^{2}(\Omega)}.
\end{eqnarray}
Now, for $u\in H^{2}_{0}(\Omega)$, we have
\begin{eqnarray}\label{w2}
\|\nabla u\|_{\mathrm{L}^{2}(\Omega)}&=&[(-\triangle u,u)_{\mathrm{L}^{2}(\Omega)}]^{1/2}\nonumber\\
&\leq &\|\triangle u\|_{\mathrm{L}^{2}(\Omega)}^{1/2} \;\| u\|_{\mathrm{L}^{2}(\Omega)}^{1/2} \nonumber\\
&\leq &\frac{\epsilon^{2}}{2}\|\triangle u\|_{\mathrm{L}^{2}(\Omega)}+\frac{1}{2\epsilon^{2}}\|u\|_{\mathrm{L}^{2}(\Omega)}\,,
\end{eqnarray}
for all $\epsilon>0$. So, by  inequalities~\eqref{w1} and~\eqref{w2}, we have
\begin{equation}\label{w3}
\|\triangle u\|_{\mathrm{L}^{2}(\Omega)}\leq \frac{1}{1-C_{1}\epsilon^{2}/2}\|H^{A}(u)\|_{\mathrm{L}^{2}(\Omega)}+\frac{C_{0}+C_{1}/(2\epsilon^{2})}{1-C_{1}\epsilon^{2}/2}\|u\|_{\mathrm{L}^{2}(\Omega)}.
\end{equation}

 By Poincar\'e's Inequality, the norm of $H^{2}(\Omega)$ is equivalent, in~$H^{2}_{0}(\Omega)$, to
\begin{equation} 
 \Big(\|\cdot\|_{\mathrm{L}^{2}(\Omega)}^{2}+\sum_{|\alpha|=2}\|\partial^{\alpha}(\cdot)\|_{\mathrm{L}^{2}(\Omega)}^{2}\Big)^{1/2}\,;
 \end{equation}
  for all $f\in H^{2}_{0}(\Omega)$, we have,
\begin{eqnarray*}
\sum_{|\alpha|=2}\|\partial^{\alpha}f\|_{\mathrm{L}^{2}(\Omega)}^{2}&=&\sum_{|\alpha|=2}(\partial^{\alpha}f,\partial^{\alpha}f)_{\mathrm{L}^{2}(\Omega)}=\sum_{i,j=1}^{2}(\partial_{i}\partial_{j}f,\partial_{i}\partial_{j}f)_{\mathrm{L}^{2}(\Omega)}\\&=&\sum_{i,j=1}^{2}(\partial_{i}^{2}f,\partial_{j}^{2}f)_{\mathrm{L}^{2}(\Omega)}=\|\triangle f\|_{\mathrm{L}^{2}(\Omega)}^{2}\,,
\end{eqnarray*}
where in the last equality we have performed a integration by parts, so, we get
\begin{equation}\label{w4}
\|u\|_{H^{2}(\Omega)}\approx(\|u\|_{\mathrm{L}^{2}(\Omega)}^{2}+\|\triangle u\|_{\mathrm{L}^{2}(\Omega)}^{2})^{1/2}
\end{equation}
 for all $u\in H^{2}_{0}(\Omega)$. Therefore, by~\eqref{w3} and~\eqref{w4} we obtain  $\|u\|_{H^{2}(\Omega)}\leq C' \|u\|_{A}$ for all $u\in H^{2}_{0}(\Omega)$, and so $\|\cdot\|_{A}\approx\|\cdot\|_{H^{2}(\Omega)}$ in $H^{2}_{0}(\Omega)$.
From these facts, we obtain
\begin{equation*}
\mathrm{Dom}\,H^{A}_{\mathrm{min}}=\mathrm{Dom}\,\overline{H^{A}_{0}}=\overline{\mathrm{Dom}\,H^{A}_{0}}^{\|\cdot\|_{A}}=\overline{\mathrm{Dom}\,H^{A}_{0}}^{\|\cdot\|_{H^{2}(\Omega)}}=\overline{\mathrm{C}^{\infty}_{0}(\Omega)}^{\|\cdot\|_{H^{2}(\Omega)}}=H^{2}_{0}(\Omega)\,.
\end{equation*}

\end{proof}

Now we discuss a first version of an integration by parts formula associated with the operator~$H^{A}$.

\begin{lemma}
 Let~$\Omega$ be a Lipschitz domain with compact boundary, $u\in H^{2}(\Omega)$ and $v\in H^{1}(\Omega)$; if
\begin{equation*}
\Phi_{A}(u,v)=(\nabla_{A}u,\nabla_{A}v)_{\mathrm{L}^{2}(\Omega)^{n}}\,,
 \end{equation*}
then
 \begin{equation}\label{intp}
\Phi_{A}(u,v)=(H^{A}u,v)_{\mathrm{L}^{2}(\Omega)}+(\gamma_{\mathrm{N}}^{A} u,\gamma_{\mathrm{D}}v)_{\mathrm{L}^{2}(\partial\Omega)}\,.
 \end{equation}
\end{lemma}
\begin{proof}
It is enough to consider $u,v\in \mathrm{C}^{\infty}_{0}(\overline{\Omega})$ since the general case follows by a density argument. Note that,
\begin{equation*}
H^{A}u=-\triangle u-2i{A}\cdot\nabla u+(|{A}|^{2}-i\,\mathrm{div}{A})u\,;
\end{equation*}
on the other hand,
\begin{eqnarray*}
\Phi_{A}(u,v)&=&\int_{\Omega}\mathrm{d}^{n}x \, \big(\overline{\nabla u}\cdot\nabla v+i \overline{\nabla u}\cdot{A}v-i{A}\cdot\overline{u}\nabla v+|{A}|^{2}\overline{u}v\big) \\
&=&(\nabla u,\nabla v)_{\mathrm{L}^{2}(\Omega)}+i\int_{\Omega}\mathrm{d}^{n}x \,(\overline{\nabla u}\cdot{A}v-{A}\cdot\overline{u}\nabla v)+\int_{\Omega}\mathrm{d}^{n}x \, |{A}|^{2}\overline{u}v \\
&=&-(\triangle u,v)_{\mathrm{L}^{2}(\Omega)}+(\gamma_{\mathrm{N}}u,\gamma_{\mathrm{D}}v)_{\mathrm{L}^{2}(\partial\Omega)}\\
&+&i\int_{\Omega}\mathrm{d}^{n}x \,(\overline{\nabla u}\cdot{A}v-{A}\cdot\overline{u}\nabla v)+\int_{\Omega}\mathrm{d}^{n}x \, |{A}|^{2}\overline{u}v\\
&=&-(\triangle u,v)_{\mathrm{L}^{2}(\Omega)}+(\gamma_{\mathrm{N}}u,\gamma_{\mathrm{D}}v)_{\mathrm{L}^{2}(\partial\Omega)}\\
&+&2i\int_{\Omega}\mathrm{d}^{n}x \,v\overline{\nabla u}\cdot{A}+i\int_{\Omega}\mathrm{d}^{n}x \, v\overline{u}\, \mathrm{div}{A} \\ &-&i\int_{\partial\Omega}\mathrm{d}^{n-1}\omega\; \overline{u}v{A}\cdot\nu+\int_{\Omega}\mathrm{d}^{n}x \, |{A}|^{2}\overline{u}v \\
&=&(H^{A}u,v)_{\mathrm{L}^{2}(\Omega)}+(\gamma_{\mathrm{N}}^{A} u,\gamma_{\mathrm{D}}v)_{\mathrm{L}^{2}(\partial\Omega)}\,,
\end{eqnarray*}
where in the fourth equality we have used that
\[
\overline{u}{A}\cdot\nabla v=-v\,\overline{\nabla u}\cdot{A}-v\,\overline{u}\,\mathrm{div}{A}+\mathrm{div}(v\,\overline{u}{A})
\] and the  Theorem~3.34 of~\cite{STS}, which is a version of the divergence theorem for Lipschitz Domains.
\end{proof}

 The  next  lemma will be important to obtain some generalization of this integration by parts formula.  It is a particular case of  Theorem~6.9 in~\cite{SRAC}. A result in this direction appears in~\cite{BVP}, where the author considers the case $s=1$, $H^{A}=\triangle$  (i.e., $A=0$)  and~$\Omega$ bounded; in~\cite{URD}, using an approach similar to~\cite{BVP}, the authors have generalized to the case $s<2$.

\begin{lemma}\label{denso}  
Let  $\Omega\subset \mathbb R^{n}$ be a Lipschitz domain with compact boundary.  Then, for $ 0\leq s< 2$, $\mathrm{C}^{\infty}_{0}(\overline{\Omega})$ is dense in $H^{s}(A,\Omega)=\left\{u\in H^{s}(\Omega)|\; H^{A}u \in \mathrm{L}^{2}(\Omega)\right\}$, when equipped with the norm $\|\cdot\|_{A,s}=\|\cdot\|_{H^{s}(\Omega)}+\|H^{A}(\cdot)\|_{\mathrm{L}^{2}(\Omega)}$. 
\end{lemma}
\begin{proof}
 This result is a  consequence of the fact that ${A}\in(\textrm{W}^{1}_{\infty}(\overline{\Omega}))^{n}$ and a direct application of  Theorem~6.9 of~\cite{SRAC}.
\end{proof}

 By  Lemma~\ref{denso} we can further extend the operator~$\gamma^{A}_{N}$ defined in~\eqref{TN} in the following way:
 \begin{equation*}
 {\tilde\gamma^{A}_{N}}:H^{1}(A,\Omega)\rightarrow H^{-\frac{1}{2}}(\partial\Omega)\,,
 \end{equation*}
 where, for $u\in H^{1}(A,\Omega)$, we have that $ {\tilde\gamma^{A}_{N}}u \in H^{-\frac{1}{2}}(\partial\Omega)$ is defined as
\begin{equation}\label{intn}
\langle g,{\tilde\gamma^{A}_{N}}u\rangle_{\frac{1}{2}}=\Phi_{A}(u,G)-\langle l(H^{A}u),G\rangle_{1},\quad g\in H^{\frac{1}{2}}(\partial\Omega)\,,
\end{equation}
where $G\in H^{1}(\Omega)$ is such that $\gamma_{\mathrm{D}} G=g$ and $\|G\|_{H^{1}(\Omega)}\leq c\,\|g\|_{H^{\frac{1}{2}}(\partial\Omega)}$ and $l:\mathrm{L}^{2}(\Omega)\rightarrow H^{-1}(\Omega)$ is the natural inclusion.
To see that this definition makes sense, it is enough to show that it does not depend on the particular choice of~$G$ that satisfies the above condition. By linearity, this is equivalent to: if $G\in H^{1}(\Omega)$ is such that $\gamma_{\mathrm{D}} G=0$, then $\Phi_{A}(u,G)-\langle l(H^{A}u),G\rangle_{1}=0$; however this follows from the fact that this holds true if $u\in \mathrm{C}^{\infty}_{0}(\overline{\Omega})$ since $\mathrm{C}^{\infty}_{0}(\overline{\Omega})$ is dense in  $H^{1}(A,\Omega)$.

\subsection{Magnetic Dirichlet and Neumann realizations}

In what follows we introduce the  Dirichlet~$H^{A}_{\mathrm D}$ and the  Neumann~$H^{A}_{\mathrm N}$  magnetic self-adjoint realizations, two operators that will play a fundamental role ahead. 

\begin{lemma}
Let~$\Omega$ be a Lipschitz domain with compact boundary. Then there exist  real numbers $c>0$ and $C>0$ such that, for all $u\in H^{1}(\Omega)$,
\begin{equation}\label{coer}
|\Phi_{A}(u,u)|\geq c\,\|u\|_{H^{1}(\Omega)}^{2}-C\,\|u\|_{\mathrm{L}^{2}(\Omega)}^{2}.
\end{equation}

\end{lemma}
\begin{proof} We have 
\begin{eqnarray*}
\Phi_{A}(u,u)&=&\|\nabla u\|_{\mathrm{L}^{2}(\Omega)}^{2}+2 \int_{\Omega}\mathrm{d}^{n-1}x\, \Imm\overline{u}{A}\cdot\nabla u+\int_{\Omega}\mathrm{d}^{n-1}x\, |{A}|^{2}|u|^{2}\\
&\geq &\|\nabla u\|_{\mathrm{L}^{2}(\Omega)}^{2}+\int_{\Omega}\mathrm{d}^{n-1}x\, |{A}|^{2}|u|^{2}-2\int_{\Omega}\mathrm{d}^{n-1}x\, |\Imm\overline{u}{A}\cdot\nabla u|\\
&\geq &(1-\epsilon^{2})\|\nabla u\|_{\mathrm{L}^{2}(\Omega)}^{2}+(1-\frac{1}{\epsilon^{2}})\int_{\Omega}\mathrm{d}^{n-1}x\, |{A}|^{2}|u|^{2}\\
&\geq &(1-\epsilon^{2})\|\nabla u\|_{\mathrm{L}^{2}(\Omega)}^{2}-|(1-\frac{1}{\epsilon^{2}})||\,|{A}|^{2}||_{L^{\infty}(\Omega)}\int_{\Omega}\mathrm{d}^{n-1}x\, |u|^{2}\,;
\end{eqnarray*}
 in the second inequality we have used that
\[
|\Imm\overline{u}\,{A}\cdot\nabla u|\leq  \frac12\Big({\frac{1}{\epsilon^{2}}|u\,A|^{2}+\epsilon^{2}|\nabla u|^{2}}\Big),
\]
and from this  the  statement  the lemma follows. 
\end{proof}
Consider the operators~$H^{A}_{\mathrm D}$, $H^{A}_{\mathrm N}$ given by
\begin{equation*}
\mathrm{Dom}\,H^{A}_{\mathrm D}=\left\{u\in H^{1}(\Omega) |\; H^{A}u\in \mathrm{L}^{2}(\Omega),\;\gamma_{\mathrm{D}}u=0\;\mathrm{in}\; H^{\frac{1}{2}}(\partial\Omega)\right\},\quad H^{A}_{\mathrm D}u:=H^{A}u
\end{equation*}
and
\begin{equation*}
\mathrm{Dom}\,H^{A}_{\mathrm N}=\left\{u\in H^{1}(\Omega) |\; H^{A}u\in \mathrm{L}^{2}(\Omega),\; {\tilde\gamma^{A}_{\mathrm{N}}}u=0\;\mathrm{in}\; H^{-\frac{1}{2}}(\partial\Omega)\right\},\quad H^{A}_{\mathrm N}u:=H^{A}u\,.
\end{equation*}

\begin{proposition}\label{opdirneu}
If~$\Omega$ is a Lipschitz domain with compact boundary, then the operators~$H^{A}_{\mathrm D}$ and $H^{A}_{\mathrm N}$ are self-adjoint.
\end{proposition}
\begin{proof} Consider first  the operator~$H^{A}_{\mathrm D}$. Let $\Phi_{A,\mathrm{D}}$ be the following sesquilinear form,
\begin{equation*}
\Phi_{A,\mathrm{D}}(u,v)=\Phi_{A}(u,v),\quad \mathrm{Dom}\,\Phi_{A,\mathrm{D}}=H^{1}_{0}(\Omega);
\end{equation*}
by~\eqref{coer} one can conclude that~$\Phi_{A,\mathrm{D}}$ is closed and so the operator ${\tilde H^{A}_{\mathrm D}}$ defined by
\begin{eqnarray*}
\mathrm{Dom}\,{\tilde H^{A}_{\mathrm D}}&=&\big\{u\in H^{1}_{0}(\Omega)| \;\exists w_{u} \in \mathrm{L}^{2}(\Omega)\;\mathrm{so\;that} \;
\Phi_{A,\mathrm{D}}(v,u)=(v,w_{u})_{\mathrm{L}^{2}(\Omega)},\; \forall v\in H^{1}_{0}(\Omega)\big\}\,, \\
{\tilde H^{A}_{\mathrm D}}u&:=&w_{u}\,,
\end{eqnarray*}
is self-adjoint. To conclude the proof we  show that $H^{A}_{\mathrm D}={\tilde H^{A}_{\mathrm D}}$.

Let  $v\in \mathrm{Dom}\,{\tilde H^{A}_{\mathrm D}}$, then since $\mathrm{C}^{\infty}_{0}(\Omega)\subset H^{1}_{0}(\Omega)$ we have
\begin{equation*}
\int_{\Omega}\mathrm{d}^{n}x \,\overline{u}\,w_{v}=\Phi_{A,\mathrm{D}}(u,v)=\int_{\Omega}\mathrm{d}^{n}x \,\overline{H^{A}u}\,v,\quad \forall u\in \mathrm{C}^{\infty}(\Omega),
\end{equation*}
  where the last equality follows by~\eqref{intp}. Thus, $w_{v}=H^{A}v$ in the sense of distributions, in particular $H^{A}v\in \mathrm{L}^{2}(\Omega)$, and so ${\tilde H^{A}_{\mathrm D}}\subset H^{A}_{\mathrm D}$.

Now, let  $v\in\mathrm{Dom}\,H^{A}_{\mathrm D}$; by taking $w_{v}:= H^{A}v\in \mathrm{L}^{2}(\Omega)$ and using~\eqref{intn} we have
\begin{equation*}
\int_{\Omega}\mathrm{d}^{n}x \,\overline{u}\,w_{v}=\int_{\Omega}\mathrm{d}^{n}x \,\overline{u}\,H^{A}v=\Phi_{A,\mathrm{D}}(u,v)
\end{equation*}
for all $u\in H^{1}_{0}(\Omega)$, and so  $ H^{A}_{\mathrm D}\subset{\tilde H^{A}_{\mathrm D}}$, which concludes the proof that~$H^{A}_{\mathrm D}$ is self-adjoint.

Now we address $H^{A}_{\mathrm N}$. Let $\Phi_{A,\mathrm{N}}$ be  the following sesquilinear form
\begin{equation*}
\Phi_{A,\mathrm{N}}(u,v)=\Phi_{A}(u,v),\quad \mathrm{Dom}\,\Phi_{A,\mathrm{N}}=H^{1}(\Omega)\,.
\end{equation*}
By using equation~\eqref{coer}, and the fact that $\Phi_{A,\mathrm{N}}$ is nonnegative (so bounded from below), we  conclude that~$\Phi_{A,\mathrm{N}}$ is closed and so ${\tilde H^{A}_{\mathrm N}}$, given by
\begin{eqnarray*}
\mathrm{Dom}\,{\tilde H^{A}_{\mathrm N}}&=&\big\{u\in H^{1}(\Omega)| \;\exists w_{u} \in \mathrm{L}^{2}(\Omega)\;\mathrm{so\;that}\;\Phi_{A,\mathrm{N}}(v,u)=(v,w_{u})_{\mathrm{L}^{2}(\Omega)},\; \forall v\in H^{1}(\Omega)\big\}, \\
{\tilde H^{A}_{\mathrm N}}u&=&w_{u},
\end{eqnarray*}
is self-adjoint. To conclude we  show that $H^{A}_{\mathrm N}={\tilde H^{A}_{\mathrm N}}$.

Let  $v\in \mathrm{Dom}\,{\tilde H^{A}_{\mathrm N}}$, then since $\mathrm{C}^{\infty}_{0}(\Omega)\subset H^{1}(\Omega)$, we have
\begin{equation*}
\int_{\Omega}\mathrm{d}^{n}x \,\overline{u}\,w_{v}=\Phi_{A,\mathrm{N}}(u,v)=\int_{\Omega}\mathrm{d}^{n}x \,\overline{H^{A}u}\,v
\end{equation*}
for all $u\in \mathrm{C}^{\infty}_{0}(\Omega)$, where the last equality follows by~\eqref{intp}. Thus, we have $w_{v}=H^{A}v$ in the sense of distributions, in particular $H^{A}v\in \mathrm{L}^{2}(\Omega)$. Furthermore, by~\eqref{intn} we have, for all $u\in H^{1}(\Omega)$,
\begin{eqnarray*}
\Phi_{A,\mathrm{N}}(u,v)&=&(u,H^{A}v)_{\mathrm{L}^{2}(\Omega)}+\overline{\langle{\tilde\gamma_{\mathrm{D}}u,\gamma_{\mathrm{N}}}v\rangle}_{\frac{1}{2}}\\
&=&(u,w_{v})_{\mathrm{L}^{2}(\Omega)}+\overline{\langle{\gamma_{\mathrm{D}}u,\tilde\gamma_{\mathrm{N}}}v\rangle}_{\frac{1}{2}}\\
&=&\Phi_{A,\mathrm{N}}(u,v)+\overline{\langle{\tilde\gamma_{\mathrm{D}}u,\gamma_{\mathrm{N}}}v\rangle}_{\frac{1}{2}}
\end{eqnarray*}
and so ${\tilde\gamma_{\mathrm{N}}}v=0$ since $\gamma_{\mathrm{D}}:H^{1}(\Omega)\rightarrow H^{\frac{1}{2}}(\Omega)$ is onto. Hence ${\tilde H^{A}_{\mathrm N}}\subset H^{A}_{\mathrm N}$.

On the other hand, if $u\in \mathrm{Dom}\,H^{A}_{\mathrm N}$ then, analogously to the case of the operator~$H^{A}_{\mathrm D}$, one can show that $u\in\mathrm{Dom}\,{\tilde H^{A}_{\mathrm N}}$, and so $H^{A}_{\mathrm N}={\tilde H^{A}_{\mathrm N}}$, that is, the operator $H^{A}_{\mathrm N}$ is self-adjoint.
\end{proof}

\begin{corollary}
 Let $\Omega$ be a bounded Lipschitz domain, then $ H^{A}_{\mathrm D}$ is an operator with discrete spectrum. 
\end{corollary}
\begin{proof}
Indeed, by the Theorem of Lax-Milgram and  equation~\eqref{coer}, we can prove that the problem below  has a unique solution  $u\in H^{1}_{0}(\Omega)$ and we have $\|u\|_{H^{1}(\Omega)}\leq C' \|f\|_{\mathrm{L}^{2}(\Omega)}$ for some $C'>0$,
\begin{eqnarray*}
(H^{A}+C)u&=&f, \quad f\in \mathrm{L}^{2}(\Omega),\\
\gamma_{\mathrm{D}}u&=&0,
\end{eqnarray*}
where $C$ is the constant in the inequality \eqref{coer}.
So, $(H^{A}_{\mathrm D}+C)^{-1} :\mathrm{L}^{2}(\Omega)\rightarrow H^{1}_{0}(\Omega)$ is continuous and  since the inclusion $H^{1}_{0}(\Omega)\hookrightarrow \mathrm{L}^{2}(\Omega)$ is compact if $\Omega$ is bounded, the operator $(H^{A}_{\mathrm D}+C)^{-1}  :\mathrm{L}^{2}(\Omega)\rightarrow \mathrm{L}^{2}(\Omega)$ is compact, thus, $H^{A}_{\mathrm D}$ is discrete.
\end{proof}

\section{Dirichlet and Neumann traces over the maximal domain}\label{cap4}

We  review some facts on traces discussed in~\cite{ADO} and present some generalizations  to the situation with magnetic potential. At the end  we   recall the concept of quasi-convex domain and we will introduce the concept of magnetic Neumann  trace; this will be an important step for the construction of boundary triples for the maximal magnetic operator.

\subsection{Basic facts}

\begin{theorem}\label{gamma2}
Let~$\Omega$ be a Lipschitz domain with compact boundary, and denote by~$\nu$ the unit vector field normal to its boundary $\partial\Omega$. Let \[
\mathcal F := \big\{(g_{0},g_{1})\in H^{1}(\partial\Omega)\dot{+}\mathrm{L}^{2}(\partial\Omega,\mathrm{d}^{n-1}\omega)\mid\;\nabla_{\mathrm{tan}}g_0+g_{1}\nu\in H^{\frac{1}{2}}(\partial\Omega)^{n}\big\}
\]  be equipped with the norm
\begin{equation*}
\|((g_{0},g_{1})\|_{\partial\Omega}=\|g_{0}\|_{H^{1}(\partial\Omega)}+\|g_{1}\|_{\mathrm{L}^{2}(\partial\Omega,\mathrm{d}^{n-1}\omega)}+\|\nabla_{\mathrm{tan}}g_0+g_{1}\nu\|_{ H^{\frac{1}{2}}(\partial\Omega)^{n}}.
\end{equation*}
Then, the operator
\begin{equation*}
\gamma_{2}:H^{2}(\Omega)\rightarrow\mathcal F,\quad \gamma_{2}u:=(\gamma_{\mathrm{D}}u,\gamma_{\mathrm{N}}u)\,,
\end{equation*}
is well defined, linear, bounded and has right inverse that is bounded. Furthermore, the kernel of $\gamma_{2}$ is $H^{2}_{0}(\Omega)$.
\end{theorem}
\begin{proof} This is an easy adaptation of the proof for  bounded domains presented in~\cite{ADO}. Fix an open ball~$B_{r}$ of radius~$r>0$ such that $\partial\Omega\subset B_{r/2}$ and a function $\varphi\in \mathrm{C}^{\infty}_{0}(\mathbb{R}^{n})$ such that $\varphi(x)=1,$ for all $x\in  B_{r/2}$, and $\varphi(x)=0,$ for all $x\in  \mathbb{R}^{n}\setminus B_{3r/4}$. Denote by $\gamma_{2}^{L},\; L=\Omega$ or $\Omega\cap B_{r}$, the application $\gamma_{2}$ with domain $H^{2}(\Omega)$ or $H^{2}(\Omega\cap B_{r})$, respectively. Note that, for all $u\in \mathrm{C}^{\infty}_{0}(\overline{\Omega})$ one has $ \gamma_{2}^{\Omega}(u)=\gamma_{2}^{\Omega\cap B_{r}}(\varphi u)$. Hence
\[
\|\gamma_{2}^{\Omega}(u)\|_{\partial\Omega}=\|\gamma_{2}^{\Omega\cap B_{r}}(\varphi u)\|_{\partial(\Omega\cap B_{r})}\leq C\,\|\varphi u\|_{H^{2}(\Omega\cap B_{r})}\leq C'\, \|u\|_{H^{2}(\Omega)}\,.
\]
Since $ \mathrm{C}^{\infty}_{0}(\overline{\Omega})$ is dense in $H^{2}(\Omega)$, from this inequality it follows that $\gamma_{2}$, with domain  $\mathrm{C}^{\infty}_{0}(\overline{\Omega})$, can be extended continuously, and in a unique way, to an application from $H^{2}(\Omega)$ to  $\mathcal F$. To see that this extension has a bounded right inverse, it is enough to note that if $\zeta$ is an inverse of $\gamma_{2}^{\Omega\cap B_{r}}$, then $E\circ\zeta$ is the required inverse of~$\gamma_{2}^{\Omega}$, where $E$ is a continuous extension operator from $H^{2}(\Omega\cap B_{r})$ to $H^{2}(\Omega)$. The last statement of the theorem is an easy consequence of the above construction  and the bounded case discussed in~\cite{ADO}.
\end{proof}

\begin{definition}[\cite{ADO}]
Let~$\Omega$ be a Lipschitz domain with compact boundary. The space $N^{\frac{1}{2}}(\partial\Omega)$ is defined by
\begin{equation*}
N^{\frac{1}{2}}(\partial\Omega):=\left\{g\in \mathrm{L}^{2}(\partial\Omega,\mathrm{d}^{n-1}\omega)\mid\; g\nu_{i}\in H^{\frac{1}{2}}(\partial\Omega),\; 1\leq i\leq n\right\}
\end{equation*}
 and equipped with the norm
 \begin{equation}\label{norn1/2}
 \|g\|_{N^{\frac{1}{2}}(\partial\Omega)}=\big(\sum_{i=1}^{n}\|g\nu_{i}\|_{H^{\frac{1}{2}}(\partial\Omega)}^2\big)^{1/2}.
 \end{equation}
 \end{definition}
 
 The norm~\eqref{norn1/2} is clearly obtained from the inner product
\begin{equation}\label{inner1/2}
 \langle u,v \rangle_{N^{\frac{1}{2}}(\partial\Omega)}=\sum_{i=1}^{n}\langle \nu_{i}u,\nu_{i}v\rangle_{H^{\frac{1}{2}}(\partial\Omega)},
  \end{equation}
 where $\langle u,v \rangle_{H^{\frac{1}{2}}(\partial\Omega)}$ is an inner product in the Hilbert space $H^{\frac{1}{2}}(\partial\Omega)$. This space is a reflexive Banach space that can be  continuously embedded in $ \mathrm{L}^{2}(\partial\Omega,\mathrm{d}^{n-1}\omega)$. Moreover, if~$\Omega$ is bounded and $\partial\Omega\in \mathrm{C}^{1,r}$ with $r>1/2$, we have $N^{\frac{1}{2}}(\partial\Omega)=H^{\frac{1}{2}}(\partial\Omega)$ with equivalent norms; see  Lemma~6.2 of~\cite{ADO}.
 
The following result is a direct consequence of Lemma~6.3 of~\cite{ADO}.
 
 \begin{lemma}\label{tnh2h10}
  Let~$\Omega$ be a Lipschitz domain with compact boundary, then
 \begin{equation*}
 \gamma_{\mathrm{N}}^{A}:H^{2}(\Omega)\cap H^{1}_{0}(\Omega)\rightarrow N^{\frac{1}{2}}(\partial\Omega)
 \end{equation*}
 is well defined, linear, bounded, onto and has a bounded right inverse, moreover its kernel is~$H^{2}_{0}(\Omega)$.
 \end{lemma}
 
 The next result extends the definition of $\gamma_{\mathrm{D}}$ to the domain of~$H^{A}_{\mathrm{max}}$.

\begin{theorem}\label{gammaN1/2}
Let~$\Omega$ be a Lipschitz domain with compact boundary. Then there exists a unique extension ${\hat\gamma_{\mathrm{D}}}$ of $\gamma_{\mathrm{D}}$,
\begin{equation}\label{gammadmax}
{\hat\gamma_{\mathrm{D}}}:\mathrm{Dom}\,H^{A}_{\mathrm{max}}=\left\{u\in \mathrm{L}^{2}(\Omega)|\; H^{A}u\in \mathrm{L}^{2}(\Omega)\right\}\rightarrow(N^{\frac{1}{2}}(\partial\Omega))^{*},
\end{equation}
that is compatible with~\eqref{trd} in the following sense: for $3/2\geq s>1/2$  and for all $ u \in H^{s}(\Omega)$ with $H^{A}u\in \mathrm{L}^{2}(\Omega)$, one has  ${\hat\gamma_{\mathrm{D}}}u=\gamma_{\mathrm{D}}u$. Furthermore, the range of ${\hat\gamma_{\mathrm{D}}}$ is dense in $(N^{\frac{1}{2}}(\partial\Omega))^{*}$ and the following integration by parts formula holds,
\begin{equation}\label{intp2}
\langle\gamma_{\mathrm{N}}^{A}w,{\hat\gamma_{\mathrm{D}}}u\rangle_{N^{\frac{1}{2}}(\partial\Omega)}=-\left((H^{A}w,u)_{\mathrm{L}^{2}(\Omega)}-(w,H^{A}u)_{\mathrm{L}^{2}(\Omega)}\right),
\end{equation}
where $w\in H^{2}(\Omega)\cap H^{1}_{0}(\Omega)$ and $ u\in \mathrm{L}^{2}(\Omega)$ with $H^{A}u\in \mathrm{L}^{2}(\Omega)$, and
\[
\langle \cdot,\cdot\rangle_{N^{\frac{1}{2}}(\partial\Omega)}:N^{\frac{1}{2}}(\partial\Omega)\times(N^{\frac{1}{2}}(\partial\Omega))^{*}\rightarrow \mathbb C
\]  represents the pairing between a vector and a linear functional on $N^{\frac{1}{2}}(\partial\Omega)$.
\end{theorem}
\begin{proof}
 Take $u\in \mathrm{Dom}\,H^{A}_{\mathrm{max}}$, and define ${\hat\gamma_{\mathrm{D}}}u\in (N^{\frac{1}{2}}(\partial\Omega))^{*}$ in the following way: take $g\in N^{\frac{1}{2}}(\partial\Omega)$; by Lemma~\ref{tnh2h10} there exists $w\in  H^{2}(\Omega)\cap H^{1}_{0}(\Omega)$ such that $\gamma_{\mathrm{N}}^{A}(w)=g$ and $\|w\|_{H^{2}(\Omega)}\leq C \|g\|_{N^{\frac{1}{2}}(\partial\Omega)}$.
 Define, then, ${\hat\gamma_{\mathrm{D}}}u \in (N^{\frac{1}{2}}(\partial\Omega))^{*}$ through
\begin{equation*}
 \langle g,{\hat\gamma_{\mathrm{D}}}u\rangle_{N^{\frac{1}{2}}(\partial\Omega)}:=-\left((H^{A}w,u)_{\mathrm{L}^{2}(\Omega)}-(w,H^{A}u)_{\mathrm{L}^{2}(\Omega)}\right).
\end{equation*}  
 
 To see that ${\hat\gamma_{\mathrm{D}}}u $ is well defined by this relation, we need to show that the above definition does not depend on the particular choice of $w$ satisfying the above conditions,  which is equivalent to:
if $w\in  H^{2}(\Omega)\cap H^{1}_{0}(\Omega)$  satisfies  $\gamma_{\mathrm{N}}^{A}(w)=0$, then  $(H^{A}w,u)_{\mathrm{L}^{2}(\Omega)}=(w,H^{A}u)_{\mathrm{L}^{2}(\Omega)}$. Note that in this case, by the Lemma~\ref{tnh2h10}, $w\in \mathrm{ker}\gamma_{\mathrm{N}}^{A}= H^{2}_{0}(\Omega)$. If $w\in \mathrm{C}^{\infty}_{0}(\Omega)$ and $u\in  \mathrm{C}^{\infty}_{0}(\overline{\Omega})$, this equality holds; the general case follows from this since, $ \mathrm{C}^{\infty}_{0}(\Omega)$ is dense in $H^{2}_{0}(\Omega)$ and $\mathrm{C}^{\infty}_{0}(\overline{\Omega})$ is dense in $\mathrm{Dom}\,H^{A}_{\mathrm{max}}$. It is clear that such ${\hat\gamma_{\mathrm{D}}}$ is continuous and  does satisfy the required integration by parts formula. The uniqueness also follows from the fact that $\mathrm{C}^{\infty}_{0}(\overline{\Omega})$ is dense in $\mathrm{Dom}\,H^{A}_{\mathrm{max}}$.

 Now we  show that this extension is compatible with~\eqref{trd}. Pick $u \in H^{s}(\Omega)$ with $H^{A}u\in \mathrm{L}^{2}(\Omega)$, $3/2\geq s>1/2$, and consider $u_{i}\in \mathrm{C}^{\infty}_{0}(\overline{\Omega})$ such that $u_{i}\rightarrow u$ in $H^{s}(A,\Omega)$. Pick  $w\in H^{2}(\Omega)\cap H^{1}_{0}(\Omega)$ and $w_{j}\in \mathrm{C}^{\infty}_{0}(\overline{\Omega})$ such that $w_{j}\rightarrow w$ in $H^{2}(\Omega)$. Fix then~$i$ and consider the following vector field,
\begin{equation*}
G_{j}=\overline{u_{i}}\,\nabla_{A}w_{j} \in \mathrm{H}^{1}(\Omega),\quad j\in \mathbb{N};
\end{equation*}
by an easy calculation we have
\begin{eqnarray*}
\mathrm{div}\, G_{j}&=& \overline{\nabla_{A}u_i}\nabla_{A}w_{j}-\overline{u_{i}}\,H^{A}w_{j}\,,\\
\nu\cdot\gamma_{\mathrm{D}}G_{j}&=&\overline{\gamma_{\mathrm{D}}u_i}\,\gamma_{\mathrm{N}}^{A}w_{j}\,.
\end{eqnarray*} 
Thus, 
\begin{eqnarray*}
(u_{i},H^{A}w)_{\mathrm{L}^{2}(\Omega)}&=&\lim_{j\rightarrow\infty}(u_{i},H^{A}w_{j})_{\mathrm{L}^{2}(\Omega)}\\
&=&\lim_{j\rightarrow\infty}\Big\{\Big(-\int_{\Omega}\mathrm{d}^{n}x \,\;\mathrm{div}\,G_{j}\Big)+\Phi_{A}(u_{i},w_{j})\Big\}\\
&=&\lim_{j\rightarrow\infty}\Big(-\int_{\partial\Omega}\mathrm{d}^{n-1}\omega\; \overline{\gamma_{\mathrm{D}}u_{i}}\,\gamma_{\mathrm{N}}^{A}w_{j}+\Phi_{A}(u_{i},w_{j})\Big)\\
&=&\lim_{j\rightarrow\infty}\Big((H^{A}u_{i},w_{j})_{\mathrm{L}^{2}(\Omega)}+(\gamma_{\mathrm{N}}^{A}u_{i},\gamma_{\mathrm{D}}w_{j})_{\mathrm{L}^{2}(\Omega)}-\int_{\partial\Omega}\mathrm{d}^{n-1}\omega\; \overline{\gamma_{\mathrm{D}}u_{i}}\,\gamma_{\mathrm{N}}^{A}w_{j} \Big)\\
&=&(H^{A}u_{i},w)_{\mathrm{L}^{2}(\Omega)}-\int_{\partial\Omega}\mathrm{d}^{n-1}\omega\; \overline{\gamma_{\mathrm{D}}u_{i}}\,\gamma_{\mathrm{N}}^{A}w
\end{eqnarray*}
where in the last equality we have used that $\gamma_{\mathrm{D}}w_{j}\rightarrow\gamma_{\mathrm{D}}w=0$. Hence
\begin{equation*}
(u_{i},H^{A}w)_{\mathrm{L}^{2}(\Omega)}=(H^{A}u_{i},w)_{\mathrm{L}^{2}(\Omega)}-\int_{\partial\Omega}\mathrm{d}^{n-1}\omega\; \overline{\gamma_{\mathrm{D}}u_{i}}\,\gamma_{\mathrm{N}}^{A}w\,.
\end{equation*}
By taking $i\rightarrow\infty$ in the last equation we obtain
\begin{equation*}
\int_{\partial\Omega}\mathrm{d}^{n-1}\omega\; \overline{\gamma_{\mathrm{D}}u}\,\gamma_{\mathrm{N}}^{A}w=\Big((H^{A}u,w)_{\mathrm{L}^{2}(\Omega)}-(u,H^{A}w)_{\mathrm{L}^{2}(\Omega)}\Big)=\overline{\langle\gamma_{\mathrm{N}}^{A}w,{\hat\gamma_{\mathrm{D}}}u\rangle}_{N^{\frac{1}{2}}(\partial\Omega)}\,,
\end{equation*}
and since $\gamma_{\mathrm{N}}^{A}:H^{2}(\Omega)\cap H^{1}_{0}(\Omega)\rightarrow N^{\frac{1}{2}}(\partial\Omega)$ is onto, the above equation shows that $\gamma_{\mathrm{D}}u$ and ${\hat\gamma_{\mathrm{D}}}u$  coincide as antilinear functionals in $N^{\frac{1}{2}}(\partial\Omega)$, in other words $\gamma_{\mathrm{D}}u={\hat\gamma_{\mathrm{D}}}u$.

 Finally, we only need to show that ${\hat\gamma_{\mathrm{D}}}$ has a dense range; to see this we  show that $\big\{u|_{\partial\Omega}\mid\, u\in \mathrm{C}^{\infty}_{0}(\overline\Omega)\big\}$ is dense in $(N^{\frac{1}{2}}(\partial\Omega))^{*}$. Take then $\Psi$ in $((N^{\frac{1}{2}}(\partial\Omega))^{*})^{*}=N^{\frac{1}{2}}(\partial\Omega)\hookrightarrow \mathrm{L}^{2}(\partial\Omega)$ that is zero over $\big\{u|_{\partial\Omega}\mid\, u\in \mathrm{C}^{\infty}_{0}(\overline\Omega)\big\}$; to conclude it is enough to show that~$\Psi$ is zero.
By  Lemma~\ref{tnh2h10} there exists $w\in H^{2}(\Omega)\cap H^{1}_{0}(\Omega)$ such that $\gamma_{\mathrm{N}}^{A}w=\Psi$, and so
\begin{equation*}
\langle\gamma_{\mathrm{N}}^{A}w,{\hat\gamma_{\mathrm{D}}}u\rangle_{N^{\frac{1}{2}}(\partial\Omega)}=0,\quad \forall u\in \mathrm{C}^{\infty}_{0}(\overline{\Omega}).
\end{equation*}

By using this equation with~\eqref{intp2} we obtain
\begin{equation*}
(H^{A}w,u)_{\mathrm{L}^{2}(\Omega)}=(w,H^{A}u)_{\mathrm{L}^{2}(\Omega)},\quad \forall u\in \mathrm{C}^{\infty}(\overline{\Omega})\,;
\end{equation*}
on the other hand, by~\eqref{intp} we have
\begin{equation*}
(H^{A}w,u)_{\mathrm{L}^{2}(\Omega)}=(w,H^{A}u)_{\mathrm{L}^{2}(\Omega)}-(\gamma_{\mathrm{N}}^{A}w,\gamma_{\mathrm{D}}u)_{\mathrm{L}^{2}(\partial\Omega)}.
\end{equation*}
By these equations 
\begin{equation*}
(\Psi,\gamma_{\mathrm{D}}u)_{\mathrm{L}^{2}(\partial\Omega)}=(\gamma_{\mathrm{N}}^{A}w,\gamma_{\mathrm{D}}u)_{\mathrm{L}^{2}(\partial\Omega)}=0,\quad \forall u\in \mathrm{C}^{\infty}_{0}(\overline{\Omega}),
\end{equation*}
and since $\gamma_{\mathrm{D}} \mathrm{C}^{\infty}_{0}(\overline{\Omega})$ is dense in $\mathrm{L}^{2}(\partial\Omega)$ (this follows by Lemma~3.1 of~\cite{ADO} and   Lemma \ref{denso}), it then follows that $\Psi=0$, and this proves the statement.
\end{proof}

As a simple consequence of this result we recover Corollary~6.5 of \cite{ADO}.

\begin{corollary}\label{corww} Let  $\Omega$ be a Lipschitz domain with compact boundary, then each of the following inclusions is continuous and has dense range. More over, the pairing between $(N^{1/2}(\partial\Omega))^{*}$ and $N^{1/2}(\partial\Omega)$ is compatible with the inner product of $L^{2}(\Omega)$.
\begin{equation*}
N^{1/2}(\partial\Omega)\hookrightarrow \mathrm{L}^{2}(\partial\Omega)\hookrightarrow (N^{1/2}(\partial\Omega))^{*}.
\end{equation*}
\end{corollary}
\begin{proof} We will just sketch the proof.
Denote the first inclusion by $I$; since $N^{1/2}(\partial\Omega)$ is a reflexive Banach space the second inclusion is nothing else than $I^{*}$, the dual of~$I$. It is a  fact about reflexive Banach spaces that an application~$I$ is continuous if, and only if,  $I^{*}$ is continuous, and  further,  $I$ is injective and has dense range if, and only if, the same holds for~$I^{*}$. With this in mind, it is enough to prove the statement of the corollary for one inclusion. The fact that~$I$ is continuous follows directly from the definition of $N^{1/2}(\partial\Omega)$. The compatibility of the pairing and the inner product follows from the fact that the second inclusion is the dual of the first one. To see that the second inclusion has dense range, note that, by  Lemma~\ref{denso},  $\mathrm{C}^{\infty}(\overline{\Omega})$ is dense in $\mathrm{Dom}\,H^{A}_{\mathrm{max}}$, thus by the compatibility and density results of  Theorem~\ref{gammaN1/2}, the set $L=\gamma_{D}(\mathrm{C}^{\infty}(\overline{\Omega}))=\hat{\gamma}_{D}(\mathrm{C}^{\infty}(\overline{\Omega}))$ is dense in $ (N^{1/2}(\partial\Omega))^{*}$; since $L\subset L^{2}(\partial\Omega)$, the statement follows.
\end{proof}

In what follows we  introduce the space $N^{3/2}_{A}(\partial\Omega)$, which is  a generalization of the space $N^{\frac{3}{2}}(\partial\Omega)$ introduced in Section~6 of~\cite{ADO}. In  Section~\ref{cap7}~we will use the space $N^{1/2}(\partial\Omega)$ to construct a boundary triple for $H^{A}_{\mathrm{max}}$ rather the space $N^{3/2}_{A}(\partial\Omega)$, however a similar construction can be done using the space $N^{3/2}_{A}(\partial\Omega)$; see the Remark~\ref{remarkN3/2}.

\begin{definition}\label{defN32A}
Let~$\Omega$ be a Lipschitz domain with compact boundary. Define the space $N^{3/2}_{A}(\partial\Omega)$ by
\begin{eqnarray*}
N^{3/2}_{A}(\partial\Omega)&:=&\left\{g\in H^{1}(\partial\Omega)\mid\; \nabla_{\mathrm{tan}}^{A}g\in H^{\frac{1}{2}}(\partial\Omega)\right\},\quad\mathrm{where}\\
\nabla_{\mathrm{tan}}^{A}g&:=&\left(\nabla_{\mathrm{tan}}-i(\nu\cdot{A})\nu\right)g
\end{eqnarray*}
(above, we have made the abuse of notation that consists of denoting~$\gamma_{\mathrm{D}}{A}$ also by~${A}$),
and  equip  $N^{3/2}_{A}(\partial\Omega)$ with the  norm
\begin{equation*}
\|\cdot\|_{N^{3/2}_{A}(\partial\Omega)}=\|\cdot\|_{H^{1}(\partial\Omega)}+\|\nabla_{\mathrm{tan}}^{A}\cdot\|_{(H^{\frac{1}{2}}(\partial\Omega))^{n}}.
\end{equation*}
\end{definition}

\begin{lemma}
Let~$\Omega$ be a Lipschitz domain with compact boundary, then $N^{3/2}_{A}(\partial\Omega)$ is a reflexive Banach space that can be continuously embedded in $H^{1}(\partial\Omega)$. 
\end{lemma}
\begin{proof}

It is evident that the inclusion of $N^{3/2}_{A}(\partial\Omega)$ in $H^{1}(\partial\Omega)$ is continuous. To see that this  is a Banach space, take a  Cauchy sequence  $\{g_n\}_{n\in {\mathbb N}}$ in $N^{3/2}_{A}(\partial\Omega)$.  From the last statement it follows that $\{g_n\}_{n\in {\mathbb N}}$ is Cauchy in $H^{1}(\partial\Omega)$, and so it converges in this space to  $g$. Thus, $\nabla_{\mathrm{tan}}^{A}g_{n}$ converges in $\mathrm{L}^{2}(\partial\Omega)$ to $\nabla_{\mathrm{tan}}^{A}g$ and since $\nabla_{\mathrm{tan}}^{A}g_{n}$ is also  Cauchy in $ H^{\frac{1}{2}}(\partial\Omega)$, it follows that this sequence  converges to $\nabla_{\mathrm{tan}}^{A}g$ in $ H^{\frac{1}{2}}(\partial\Omega)$. Therefore, $g\in N^{3/2}_{A}(\partial\Omega)$ and $\{g_n\}_{n\in {\mathbb N}}$  converges to $g$ in $N^{3/2}_{A}(\partial\Omega)$, so, $N^{3/2}_{A}(\partial\Omega)$ is a Banach space. The reflexivity property of  $N^{3/2}_{A}(\partial\Omega)$ follows from the fact that $\Psi: g\rightarrow(g,\nabla_{\mathrm{tan}}^{A}g)$ is an isometry  between $N^{3/2}_{A}(\partial\Omega)$ and a closed subspace of  $H^{1}(\partial\Omega)\times(H^{\frac{1}{2}}(\partial\Omega))^{n}$.
\end{proof}
 
 The next lemma shows that if~$\Omega$ is bounded with smooth boundary, then $N^{3/2}_{A}(\partial\Omega)$ coincides with the ordinary space~$H^{\frac32}(\partial\Omega)$; for the definition of the space $H^{\frac32}(\partial\Omega)$ see \cite{STS} (pages 98-99).

\begin{lemma}\label{h3/2}
Let~$\Omega$ be a bounded domain of class  $\mathrm{C}^{1,r}$ with $r>1/2$, then $N^{3/2}_{A}(\partial\Omega)=N^{\frac{3}{2}}(\partial\Omega)=H^{\frac32}(\partial\Omega)$.
\end{lemma}
\begin{proof}
By Lemma~3.4 of~\cite{ADO}, 
\[
M_{\nu}:u\in H^{1/2}(\partial\Omega)\,\longrightarrow\,\nu u\in (H^{1/2}(\partial\Omega))^{n}
\] is well defined and bounded; thus,
\begin{eqnarray*}
\|M_{-i(\nu\cdot A)\nu}u\|_{(H^{1/2}(\partial\Omega))^{n}}&=&\|-i(\nu\cdot A)\nu u\|_{(H^{1/2}(\partial\Omega))^{n}}\\
&\leq& C'''\,\|(\gamma_{D}A)u\|_{(H^{1/2}(\partial\Omega))^{n}}\\
&=&\,\|(\gamma_{D}(AU)\|_{(H^{1/2}(\partial\Omega))^{n}}\\
&\leq & C''\,\|({A}U)\|_{H^{1}(\Omega)}\\
&\leq & C'\,\|u\|_{H^{1/2}(\partial\Omega)}\\
&\leq& C \,\|u\|_{H^{1}(\partial\Omega)}
\end{eqnarray*}
for all $u\in H^{1}(\partial\Omega)\hookrightarrow H^{1/2}(\partial\Omega)$, where $U\in H^{1}(\Omega)$ is such that $\gamma_{D}U=u$ and $\|U\|_{H^{1}(\Omega)} \leq K \|u\|_{H^{1/2}(\partial\Omega)}$.  That is, $M_{-i(\nu\cdot A)\nu}:H^{1}(\partial\Omega)\rightarrow(H^{1/2}(\partial\Omega))^{n}$ is bounded. Since $\nabla_{\mathrm{tan}}^{A}u=\nabla_{\mathrm{tan}}u+M_{-i(\nu\cdot A)\nu}u,$ it follows that $N^{3/2}_{A}(\partial\Omega)=N^{\frac{3}{2}}(\partial\Omega)$ and that $\|\cdot\|_{H^{1}(\partial\Omega)}+\|\nabla_{\mathrm{tan}}^{A}\cdot\|_{H^{\frac{1}{2}}(\partial\Omega)}\approx\|\cdot\|_{H^{1}(\partial\Omega)}+\|\nabla_{\mathrm{tan}}\cdot\|_{H^{\frac{1}{2}}(\partial\Omega)}$; the rest of the proof follows easily by Lemma~6.8 of~\cite{ADO}.
\end{proof}

The next lemma  presents a relationship between the space $N^{3/2}_{A}(\partial\Omega)$ and the operator~$\gamma_{\mathrm{D}}$.
\begin{lemma}\label{gammaN3/2}
Let~$\Omega$ be a Lipschitz domain with compact boundary. Then
\begin{equation*}
\gamma_{\mathrm{D}}:\left\{u\in H^{2}(\Omega)|\; \gamma_{\mathrm{N}}^{A}u=0\right\}\rightarrow N^{3/2}_{A}(\partial\Omega)
\end{equation*}
is well defined, linear, onto and has a bounded right inverse; furthermore, its kernel is~$H^{2}_{0}(\Omega)$.
\end{lemma}
\begin{proof}

Take $u\in H^{2}(\Omega)$ with $\gamma_{\mathrm{N}}^{A}u=0$; then  
\begin{equation}\label{Pororo}
\gamma_{\mathrm{N}}u=-i(\nu\cdot{A})\gamma_{\mathrm{D}}u\,,
\end{equation}
and so
\begin{equation*}
\gamma_{2}u=(\gamma_{\mathrm{D}}u,\gamma_{\mathrm{N}}u)=(\gamma_{\mathrm{D}}u,-i(\nu\cdot{A})\gamma_{\mathrm{D}}u)\,.
\end{equation*}
Therefore, by Theorem~\ref{gamma2},
\begin{equation*}
\nabla_{\mathrm{tan}}^{A}(\gamma_{\mathrm{D}}u)=\nabla_{\mathrm{tan}}(\gamma_{\mathrm{D}}u)-i(\nu\cdot{A})\nu(\gamma_{\mathrm{D}}u)=\nabla_{\mathrm{tan}}(\gamma_{\mathrm{D}}u)+(\gamma_{\mathrm{N}}u) \nu\in (H^{1/2}(\partial\Omega))^{n}
\end{equation*}
and 
\begin{eqnarray*}
\|\gamma_{\mathrm{D}}u\|_{N^{3/2}_{A}(\partial\Omega)}&=&\|\gamma_{\mathrm{D}}u\|_{H^{1}(\partial\Omega)}+\|\nabla_{\mathrm{tan}}^{A}(\gamma_{\mathrm{D}}u)\|_{(H^{\frac{1}{2}}(\partial\Omega))^{n}}\\
&=&\|\gamma_{\mathrm{D}}u\|_{H^{1}(\partial\Omega)}+\|\nabla_{\mathrm{tan}}^{A}(\gamma_{\mathrm{D}}u)+(\gamma_{\mathrm{N}}u)\|_{(H^{\frac{1}{2}}(\partial\Omega))^{n}}\\
&\leq &\|\gamma_{\mathrm{D}}u\|_{H^{1}(\partial\Omega)}+\|\gamma_{\mathrm{N}}u\|_{\mathrm{L}^{2}(\partial\Omega)}+\|\nabla_{\mathrm{tan}}^{A}(\gamma_{\mathrm{D}}u)+(\gamma_{\mathrm{N}}u)\|_{(H^{\frac{1}{2}}(\partial\Omega))^{n}}\\
&\leq & C\,\|u\|_{H^{2}(\Omega)}\,.
\end{eqnarray*}
Thus, the operator is well defined and bounded. Now we  show the existence of the bounded right inverse. Let~$\xi$ be the right inverse of~$\gamma_{2}$ given by Theorem~\ref{gamma2}, and consider the operator
\begin{eqnarray*}
l:N^{3/2}_{A}(\partial\Omega)&\rightarrow & \big\{(g_{0},g_{1})\in H^{1}(\partial\Omega)\dot{+}\mathrm{L}^{2}(\partial\Omega,\mathrm{d}^{n-1}\omega)\mid  \nabla_{\mathrm{tan}}g_0+g_{1}\nu\in H^{\frac{1}{2}}(\partial\Omega)^{n}\big\}\\
l(g)&:=&(g,-i(\nu\cdot{A})g).
\end{eqnarray*}

The operator  $l$ is bounded; in fact,
\begin{eqnarray*}
\|l(g)\|&=&\|g\|_{H^{1}(\partial\Omega)}+\|-i(\nu\cdot{A})g\|_{\mathrm{L}^{2}(\partial\Omega)}+\|\nabla_{\mathrm{tan}}^{A}g\|_{(H^{\frac{1}{2}}(\partial\Omega))^{n}}\\
&\leq &C\,(\,\|g\|_{H^{1}(\partial\Omega)}+\|\nabla_{\mathrm{tan}}^{A}g\|_{(H^{\frac{1}{2}}(\partial\Omega))^{n}}).
\end{eqnarray*}
Therefore $\xi\circ l$ is also bounded and, for all $u\in \left\{u\in H^{2}(\Omega)|\; \gamma_{\mathrm{N}}^{A}u=0\right\}$, by  equation~\eqref{Pororo},
\begin{equation*}
\gamma_{\mathrm{D}}\circ (\xi\circ l)(u)=\gamma_{\mathrm{D}}(\xi(\gamma_{\mathrm{D}}u,\gamma_{\mathrm{N}}u))=u
\end{equation*}
and
\begin{eqnarray*}
\gamma_{\mathrm{N}}^{A}\circ (\xi\circ l)(u)&=&\gamma_{\mathrm{N}}(\xi(\gamma_{\mathrm{D}}u,\gamma_{\mathrm{N}}u))-i(\nu\cdot{A})\gamma_{\mathrm{D}}(\xi(\gamma_{\mathrm{D}}u,\gamma_{\mathrm{N}}u))\\ &=&\gamma_{\mathrm{N}}u-i(\nu\cdot{A})\gamma_{\mathrm{D}}u=\gamma_{\mathrm{N}}^{A}u=0;
\end{eqnarray*}
thus, $\xi\circ l$ is the continuous right inverse that we were looking for.
\end{proof}

The next result extends the magnetic Neumann trace to the domain of the maximal operator.
\begin{theorem}\label{tnmax}
Let~$\Omega$ be a Lipschitz domain with compact boundary. Then there exists one, and only one, extension  ${\hat\gamma_{\mathrm{N}}^{A}}$ of $\gamma_{\mathrm{N}}^{A}$ with
\begin{equation}\label{tnn3/2}
{\hat\gamma_{\mathrm{N}}^{A}}:\mathrm{Dom}\,H^{A}_{\mathrm{max}}\rightarrow (N^{3/2}_{A}(\partial\Omega))^{*},
\end{equation}
that is compatible with~\eqref{TN} in the following sense: for all $s\geq3/2$ one has
\begin{equation*}
{\hat\gamma_{\mathrm{N}}^{A}}u=\gamma_{\mathrm{N}}^{A}u,\quad\forall u \in H^{s}(\Omega)\quad\mathrm{so\;\;that}\quad H^{A}u\in \mathrm{L}^{2}(\Omega).
\end{equation*}
Furthermore, this extension has dense range and 
\begin{equation}\label{intp3}
\langle{\hat\gamma_{\mathrm{N}}^{A}}u,\gamma_{\mathrm{D}}w\rangle_{N^{\frac{3}{2}}(\partial\Omega)}=-\Big((w,H^{A}u)_{\mathrm{L}^{2}(\Omega)}-(H^{A}w,u)_{\mathrm{L}^{2}(\Omega)}\Big),
\end{equation}
 for all $u\in\mathrm{Dom}\,H^{A}_{\mathrm{max}}$ and $w\in H^{2}(\Omega)$ such that $\gamma_{\mathrm{N}}^{A}w=0$, where 
 \[
 \langle\cdot,\cdot\rangle_{N^{3/2}_{A}(\partial\Omega)}:(N^{3/2}_{A}(\partial\Omega))^{*}\times N^{3/2}_{A}(\partial\Omega)\rightarrow \mathbb C
 \] represents the natural pairing between a functional and a vector in $N^{3/2}_{A}(\partial\Omega)$.
\end{theorem}

\begin{proof}

Take  $u\in\mathrm{Dom}\,H^{A}_{\mathrm{max}}$ and define  ${\hat\gamma_{\mathrm{N}}^{A}}u\in(N^{3/2}_{A}(\partial\Omega))^{*}$ in the following way: consider $g\in N^{3/2}_{A}(\partial\Omega)$, by Lemma~\ref{gammaN3/2} there exists $w\in H^{2}(\Omega)$ with $\gamma_{\mathrm{N}}^{A}w=0$ such that $\gamma_{\mathrm{D}}w=g$ and $\|w\|_{H^{2}(\Omega)}\leq C \|g\|_{N^{3/2}_{A}(\partial\Omega)}$, define then,
\begin{equation*}
\langle{\hat\gamma_{\mathrm{N}}^{A}}u,g\rangle_{N^{3/2}_{A}(\partial\Omega)}=-\Big((w,H^{A}u)_{\mathrm{L}^{2}(\Omega)}-(H^{A}w,u)_{\mathrm{L}^{2}(\Omega)}\Big)\,.
\end{equation*}
Analogously to the proof of  Lemma~\ref{gammaN1/2}, one verifies that such ${\hat\gamma_{\mathrm{N}}^{A}}$ is well defined,  unique and bounded. 

We will check now that this definition is coherent with~\eqref{TN}. For  $s\geq 3/2$, take $u\in H^{s}(\Omega)$ with $H^{A}u\in \mathrm{L}^{2}(\Omega)$. Then, for all $w\in H^{2}(\Omega)$ such that $\gamma_{\mathrm{N}}^{A}w=0$, by~\eqref{intp}, one can write
\begin{eqnarray*}
(H^{A}w,u)_{\mathrm{L}^{2}(\Omega)}&=&\Phi_{A}(w,u)-(\gamma_{\mathrm{N}}^{A}w,\gamma_{\mathrm{D}}u)_{\mathrm{L}^{2}(\Omega)}\\
&=&\Phi_{A}(w,u).
\end{eqnarray*}
On the other hand, by~\eqref{intn},
\begin{equation*}
\Phi_{A}(u,w)=(H^{A}u,w)_{\mathrm{L}^{2}(\Omega)}+(\gamma_{\mathrm{N}}^{A}u,\gamma_{\mathrm{D}}w)_{\mathrm{L}^{2}(\Omega)}.
\end{equation*}
Thus, by the  last two equations 
\begin{equation*}
(\gamma_{\mathrm{N}}^{A}u,\gamma_{\mathrm{D}}w)_{\mathrm{L}^{2}(\Omega)}=(u,H^{A}w)_{\mathrm{L}^{2}(\Omega)}-(H^{A}u,w)_{\mathrm{L}^{2}(\Omega)}=\overline{\langle{\hat\gamma_{\mathrm{N}}^{A}}u,\gamma_{\mathrm{D}}w\rangle}_{N^{\frac{3}{2}}(\partial\Omega)}
\end{equation*}
and this shows that ${\hat\gamma_{\mathrm{N}}^{A}}$,  as defined above, is coherent with~\eqref{TN}.

To show that the range of ${\hat\gamma_{\mathrm{N}}^{A}}$ is dense, it is enough to check that $\gamma_{\mathrm{N}}^{A}(\mathrm{C}^{\infty}_{0}(\overline{\Omega}))$ is dense in $(N^{\frac{3}{2}}(\partial\Omega))^{*}$. Let  $\Psi\in((N^{3/2}_{A}(\partial\Omega))^{*})^{*}=N^{3/2}_{A}(\partial\Omega)$ be an antilinear functional  that is zero on  $\gamma_{\mathrm{N}}^{A}(\mathrm{C}^{\infty}_{0}(\overline{\Omega}))$; we are going to show that $\Psi$ is equal to 0. Since $\Psi\in N^{3/2}_{A}(\partial\Omega)$, there is $w \in H^{2}(\Omega)$ with $\gamma_{\mathrm{N}}^{A}w=0$ such that $\gamma_{\mathrm{D}}w=\Psi$. Thus,
\begin{equation*}
\langle{\hat\gamma_{\mathrm{N}}^{A}}u,\gamma_{\mathrm{D}}w\rangle_{N^{3/2}_{A}(\partial\Omega)}=0,\quad\forall u\in \mathrm{C}^{\infty}_{0}(\overline{\Omega})\,.
\end{equation*}
Hence, by equation~\eqref{intp3}, one  concludes that 
\begin{equation*}
(u,H^{A}w)_{\mathrm{L}^{2}(\Omega)}=(H^{A}u,w)_{\mathrm{L}^{2}(\Omega)},\quad \forall u \in \mathrm{C}^{\infty}_{0}(\overline{\Omega}).
\end{equation*}
On the other hand, by~\eqref{intp} one obtains, for all $u \in \mathrm{C}^{\infty}_{0}(\overline{\Omega})$,
\begin{equation*}
(H^{A}w,u)_{\mathrm{L}^{2}(\Omega)}=(\gamma_{\mathrm{D}}w,\gamma_{\mathrm{N}}^{A}u)_{\mathrm{L}^{2}(\Omega)}-(\gamma_{\mathrm{N}}^{A}w,\gamma_{\mathrm{D}}u)_{\mathrm{L}^{2}(\Omega)}+(w,H^{A}u)_{\mathrm{L}^{2}(\Omega)}.
\end{equation*}
Using the fact that $\gamma_{\mathrm{N}}^{A}w=0$ and that $\Psi=\gamma_{\mathrm{D}}w$, it is found that
\begin{equation*}
\int_{\Omega}\mathrm{d}^{n}x \, \overline{\Psi}\gamma_{\mathrm{N}}^{A}u=0, \quad \forall u\in \mathrm{C}^{\infty}_{0}(\overline{\Omega}).
\end{equation*}
 Therefore, to conclude that $\Psi$ is zero it is enough to show that $\gamma_{\mathrm{N}}^{A}(\mathrm{C}^{\infty}_{0}(\overline{\Omega}))$ is dense in~$\mathrm{L}^{2}(\partial\Omega)$. But this follows by  Lemma~\ref{tnh2h10},  Corollary~\ref{corww}, and the fact that $\mathrm{C}^{\infty}_{0}(\overline{\Omega})$ is dense in $H^{2}(\Omega)\cap H^{1}_{0}(\Omega)$.
\end{proof}

\begin{corollary}
 Let~$\Omega$ be a Lipschitz domain with compact boundary. Then, we have the following inclusions
\begin{equation*}
N^{3/2}_{A}(\partial\Omega)\hookrightarrow \mathrm{L}^{2}(\partial\Omega)\hookrightarrow N^{3/2}_{A}(\partial\Omega)^{*}\,,
\end{equation*}
and both have dense ranges, further the pairing between $N^{3/2}_{A}(\partial\Omega)$ and $(N^{3/2}_{A}(\partial\Omega))^{*}$.
\end{corollary}
\begin{proof}
Denote by~$I$ the first inclusion and by $J$ the second one. Clearly the first inclusion is bounded since $N^{3/2}_{A}(\partial\Omega)\hookrightarrow H^{1}(\partial\Omega)\hookrightarrow \mathrm{L}^{2}(\partial\Omega)$, therefore the second inclusion is also bounded since  $J=I^{*}$. The image of $J$ is dense by Theorem~\ref{tnmax}. The density of the range of~$I$ follows by the injectivity of~$J$. In fact, let  $f\in \mathrm{L}^{2}(\Omega)$ be such that $f(I(N^{3/2}_{A}(\partial\Omega)))=0$ then, $Jf(N^{3/2}_{A}(\partial\Omega))=I^{*}f(N^{3/2}_{A}(\partial\Omega))=f(I(N^{3/2}_{A}(\partial\Omega)))=0$ and so $Jf=0$; since $J$ is injective, it follows that~$f=0$ and this shows that~$I$ has dense range as well.
\end{proof}

\subsection{Quasi-convex domains}\label{sectQCD}

Now we  briefly recall the concept of quasi-convex domains; for more details see the original source~\cite{ADO}, in particular Section~8. Roughly, a quasi-convex domain is a particular class of Lipschitz domain with compact boundary that, either is locally of class $\mathrm{C}^{1,r}$  or its boundary  has local convexity properties. In~\cite{ADO} the authors have shown that, for this class of domains, the functions in $\mathrm{Dom}\,\triangle_{D}$ and $\mathrm{Dom}\,\triangle_{N}$ have the $H^{2}(\Omega)$ regularity. For a general bounded Lipschitz domain, the functions in $\mathrm{Dom}\,\triangle_{D}$ have only the regularity of~$H^{\frac32}(\Omega)$; see Theorem~B.2 in~\cite{TIDP}. Using some results stated in~\cite{ADO} about the regularity of the functions in $\mathrm{Dom}\,\triangle_{D}$, we will prove that the functions in the domain of~$H^{A}_{\mathrm D}$ have the regularity of~$H^{2}(\Omega)$ as well. This result will be used to  classify all self-adjoint extensions of~$H^{A}_{\mathrm{mim}}$.

Let~$\Omega$ be a bounded Lipschitz domain in $\mathbb R^{n}$;~$\Omega$ is said to be of class $MH_{\delta}^{1/2}$, which we denote by $\partial\Omega\in MH_{\delta}^{1/2}$, if for the functions $\varphi_{i}$ in the definition of the Lipschitz domain we have  $\nabla\varphi_{i} \in (MH_{\delta}^{1/2}(\mathbb R^{n-1}))^{n}$ and  $\|\nabla\varphi\|_{(MH_{\delta}^{1/2}(\mathbb R^{n-1}))^{n}}\leq\delta$, where
\begin{equation*}
MH_{\delta}^{1/2}(\mathbb R^{n-1}):=\big\{f\in \mathrm{L}^{1}_{\mathrm{loc}}(\mathbb R^{n})\mid \;M_{f}\in B(H^{\frac12}(\mathbb R^{n}))\big\},
\end{equation*}
 $M_{f}$ is the operator of multiplication by $f$, and $B(H^{\frac12}(\mathbb R^{n}))$ is the set of bounded linear operators on~$H^{\frac12}(\mathbb R^{n})$; furthermore,
\begin{equation*}
\|f\|_{MH_{\delta}^{1/2}(\mathbb R^{n-1})}:=\|M_{f}\|_{B(H^{\frac12}(\mathbb R^{n})}.
\end{equation*}

\begin{definition}
A   Lipschitz domain with compact boundary~$\Omega$ in $\mathbb R^{m}$ is said to be almost-convex if there is a family $\{\Omega_{n}\}_{n\in {\mathbb N}}$ of open sets of $\mathbb R^{m}$ with the following properties:

i) $\partial\Omega_{n}\in \mathrm{C}^{2}$ and ${\overline{\Omega}}_{n}\subset\Omega$, for all $n\in \mathbb N$.

ii) $\Omega_{n}\nearrow\Omega$ as $n\rightarrow\infty$, in the following sense: ${\overline{\Omega}}_{n}\subset\Omega_{n+1}$ for all $n\in {\mathbb N}$ and $\bigcup_{n\in {\mathbb N}} \Omega_{n}=\Omega$.

iii) there are a neighborhood~$U$ of $\partial\Omega$ and a real function $\rho_{n}$ of class $\mathrm{C}^{2}$, for each $n\in {\mathbb N}$, defined in~$U$ such that $\rho_{n}<0$ in $U\cap\Omega_{n}$,  $\rho_{n}>0$ in $U \setminus {\overline{\Omega}}_{n}$ and $\rho_{n}=0$ in $\partial\Omega_{n}$. Furthermore, there exists $C_{1}\in (1,\infty)$ such that
\begin{equation*}
C_{1}^{-1}\leq |\nabla\rho_{n}(x)|\leq C_{1},\quad \forall x\in\partial\Omega_{n}\quad\mathrm{and}\quad \forall n\in \mathbb N.
\end{equation*}

iv) there is $C_{2}\geq 0$ such that, for all $n\in {\mathbb N}$, all  $x\in\Omega$ and all vector $\xi$ tangent to $\partial\Omega_{n}$ in~$x$, one has,
\begin{equation*}
\langle\mathrm{Hess}\rho_{n}(x)\,\xi, \xi\rangle\geq-C_{2}\,|\xi|^{2},
\end{equation*}
where $\langle\cdot,\cdot\rangle$ denote the inner product of $\mathbb R^{n}$ and $\mathrm{Hess}\rho_{n}=\left\{\frac{\partial^{2}\rho_{n}}{\partial x_{i} \partial x_{j}}\right\}_{1\leq i,j\leq n}$ is the Hessian of~$\rho_n$.
\end{definition}

Given the above considerations, we  can now  recall the definition of a quasi-convex domain.
\begin{definition}
A  Lipschitz domain with compact boundary~$\Omega$ in $\mathbb R^{m}$, with $m\geq 2$,  is said to be quasi-convex if there is $\delta>0$  small enough  such that, for $x\in\partial\Omega$, there is a neighborhood $\Omega_{x}$ of~$x$, open in~$\Omega$, such that one of the following conditions is satisfied:

i) $\Omega_{x}$ is of class $MH^{1/2}_{\delta}$ if $n\geq 3$, and of class $\mathrm{C}^{1,r}$, for some  $r\in(1/2,1)$ if $n=2$;

ii) $\Omega_{x}$ is almost-convex.
\end{definition}

The next result shows that, for a quasi-convex domain, the functions in $\mathrm{Dom}\,H^{A}_{\mathrm D}$ and $\mathrm{Dom}\,H^{A}_{\mathrm N}$, the elements of the domain of the magnetic Dirichlet and Neumann realizations, respectively, are in  $H^{2}(\Omega)$; in fact, it is a consequence of Theorem~8.11 in~\cite{ADO} and additional arguments; the main difficulties is our extension to unbounded domains and the presence of the magnetic potential (whose regularity assumptions are important here). 

\begin{theorem}
Let~$\Omega$ be a quasi-convex domain with compact boundary.  Then, both $\mathrm{Dom}\,H^{A}_{\mathrm D},\mathrm{Dom}\,H^{A}_{\mathrm N}\subset H^{2}(\Omega)$.
\end{theorem}
\begin{proof}
Assume first that~$\Omega$ is bounded.  In this case,  the result for $H^{A}_{D}$ follows directly by Theorem~8.11 of~\cite{ADO} and  the following observation: Since ${A}\in(\textrm{W}^{1}_{\infty}(\overline\Omega))^{n}$, one has $H^{A}=-\triangle+L$ with $L:H^{1}(\Omega)\rightarrow \mathrm{L}^{2}(\Omega)$. Thus, if $u\in H^{1}(\Omega)$ then $H^{A}u\in \mathrm{L}^{2}(\Omega)$ if, and only if, $\triangle u\in \mathrm{L}^{2}(\Omega)$ and the extensions of $\gamma_{\mathrm{D}}$ defined in $\mathrm{Dom}\,H^{A}_{\mathrm D}$ and $\mathrm{Dom}\,\triangle_{D}$ coincide, therefore, $\mathrm{Dom}\,H^{A}_{\mathrm D}=\mathrm{Dom}\,\triangle_{D}$.
 
 Under the assumption that $\Omega$ is bounded, we will prove the result for~$H^{A}_{N}$. The proof is similar to the proof of Theorem~8.11 of~\cite{ADO} (for the Laplacian). Take $u\in \mathrm{Dom}\,H^{A}_{\mathrm N}$, by the definition of a (bounded) quasi-convex domain and a partition of the unity argument, it is enough to show that, if $\Omega_{x}$ is as in the definition of  quasi-convex domain and $\omega\in \textrm{C}^{\infty}_{0}(\mathbb R^{n})$ is such that $\partial\Omega\cap\supp (\omega)\subset\partial\Omega\cap\partial\Omega_{x}$, then $v=(\omega u)|_{\Omega_{x}}\in H^{2}(\Omega_{x})$. Since $u\in H^{1}(\Omega)$, $\Delta u\in \mathrm{L}^{2}(\Omega)$ (which follows by the same argument from the paragraph above) we have that 
 \[
 \Delta v=[(\Delta\omega)u+2\nabla\omega\cdot\nabla u+\omega\Delta u]|_{\Omega_{x}}\in \mathrm{L}^{2}(\Omega)\,,
 \] and $v\in H^{1}(\Omega_{x})$. The conclusion of the proof in this case is divided into two cases.
 
 \
 
Case I:  Assume that $\Omega_{x}$ is of class $\textrm{C}^{1,r}$, $r>1/2$, or that $\Omega_{x}$ is of class $NH^{1/2}_{\delta}$ and $n\geq 3$. In this case the result is a consequence of the following fact that was stated in the proof of Theorem~8.11 of~\cite{ADO}:
let be $\Omega$  of class $\textrm{C}^{1,r}$ with $r>1/2$ or  $\Omega$ of class $NH^{1/2}_{\delta}$ and $n\geq 3$, if $w\in H^{1}(\Omega)$ with
\begin{equation}
\Delta w\in \mathrm{L}^{2}(\Omega),\quad \tilde{\gamma}_{N}w\in H^{1/2}(\Omega)\,,
\end{equation} 
then $w\in H^{2}(\Omega)$. Following the proof, we have, already, $v\in H^{1}(\Omega_{x})$ and $\Delta v\in \mathrm{L}^{2}(\Omega_{x})$, on the other hand,
\begin{eqnarray*}
\tilde{\gamma}_{N}^{A}v|_{\partial\Omega_{x}\cap\partial\Omega}&=&\nu\cdot\gamma_{D}(\nabla_{A}\omega u)|{\partial\Omega_{x}\cap\partial\Omega}\\
&=&[\gamma_{D}(\omega)\tilde{\gamma}_{N}^{A}u+\gamma_{D}(u)\tilde{\gamma}_{N}(\omega)|_{\partial\Omega_{x}\cap\partial\Omega}\\
&=&[\gamma_{D}(u)\tilde{\gamma}_{N}(\omega)]|_{\partial\Omega_{x}\cap\partial\Omega}\in H^{1/2}(\Omega_{x})
\end{eqnarray*}
thus,
\begin{eqnarray*}
\tilde{\gamma}_{N}v|_{\partial\Omega_{x}\cap\partial\Omega}&=&\tilde{\gamma}_{N}^{A}v|_{\partial\Omega_{x}\cap\partial\Omega}-[i\nu\cdot\gamma_{D}({A}\omega u)]|_{\partial\Omega_{x}\cap\partial\Omega}\\
&=&[\gamma_{D}(u)\tilde{\gamma}_{N}(\omega)]|_{\partial\Omega_{x}\cap\partial\Omega}-[i\nu\cdot\gamma_{D}({A}u)\gamma_{D}(\omega)]|_{\partial\Omega_{x}\cap\partial\Omega}\in H^{1/2}(\Omega_{x})
\end{eqnarray*}
since, $\gamma_{D}({A}u)\in H^{1/2}(\partial\Omega\cap\partial\Omega_{x})$ and by the fact that, under  the hypothesis of regularity of $\Omega_{x}$, assumed at the start of this case, the operator of multiplication by the components of $\nu$ maps $H^{1/2}(\Omega_{x})$ into itself in this case. Then, by the fact mentioned above, it  follows that $v\in H^{2}(\Omega_{x})$ and this finish the proof in this case.

\

Case II: $\Omega_{x}$ is almost-convex.  This case is a easy consequence of an application of the lemma 8.8 of~\cite{ADO} to the vector field $(\omega\nabla_{A}u)|_{\Omega_{x}}$, thus $v\in H^{2}(\Omega_{x})$ also in this case, and this finishes the proof of the theorem when $\Omega$ is bounded.

Consider now that~$\Omega$ is unbounded. Let us start with the case of the operator $H^{A}_{D}$, take a ball~$B_{r}$ of radius $r>0$ such that $\partial\Omega\subset B_{r/2}$, and let  $\varphi_{1},\varphi_{2}$  be a partition of the unity associated with $\{\Omega\cap B_{r},\mathbb{R}^{n}\setminus \overline{B_{r/2}}\}$. Note that since $\partial\Omega\subset B_{r/2}$, we have $\mathbb{R}^{n}\setminus \overline{B_{r/2}}\subset\Omega$. Hence, for all $u\in\mathrm{Dom}\,H^{A}_{\mathrm D}$ in~$\Omega$, $u=\varphi_{1}u+\varphi_{2}u$. Since $\varphi_{1}\in \mathrm{C}^{\infty}_{0}(B_{r})$, it follows directly  that $\varphi_{1}u\in H^{1}_{0}(\Omega\cap B_{r})$; on the other hand, $H^{A}(\varphi_{1}u)=\varphi_{1}H^{A}(u)+[ H^{A},\varphi_{1}]u$ where $[H^{A},\varphi_{1}]$, the commutator of~$H^{A}$ and the operator of multiplication by $\varphi_1$, is a differential operator of order 1 that maps $H^{1}(\Omega)$ in $L^{2}(\Omega)$. Thus, since $H^{A}u\in \mathrm{L}^{2}(\Omega)$ and $u\in H^{1}(\Omega)$, it follows that $H^{A}(\varphi_{1}u)\in \mathrm{L}^{2}(\Omega\cap B_{r})$ and so $\varphi_{1} u\in \mathrm{Dom}\,H^{A}_{\mathrm D,1}$, where $\mathrm{Dom}\,H^{A}_{\mathrm D,1}$  is the Dirichlet operator associated with~$H^{A}$ in $\Omega\cap B_{r}$. Therefore, by the case of bounded~$\Omega$, we find that $\varphi_{1} u\in H^{2}(\Omega\cap B_{r})$. On the other hand, we have  $H^{A}(\varphi_{2}u)=\varphi_{2}H^{A}(u)+[ H^{A},\varphi_{2}]u$, and since~$\varphi_{2}$  as well as its derivative of any order are bounded, it follows that $H^{A}(\varphi_{2}u)\in \mathrm{L}^{2}(\mathbb{R}^{n}\setminus \overline{B_{r/2}})$, and since $\varphi_{2}u \in H^{1}_{0}(\mathbb{R}^{n}\setminus \overline{B_{r/2}})$ we find that $\varphi_{2}u\in\mathrm{Dom}\,\triangle_{D,2}$, where  $\triangle_{D,2}$ is the Dirichlet operator associated with the Laplacian in $\mathbb{R}^{n}\setminus \overline{B_{r/2}}$.  But it is a known fact that the domain of  $\triangle_{D,2}$  is contained in $H^{2}(\mathbb{R}^{n}\setminus \overline{B_{r/2}})$; see for example Theorem~8.12 of~\cite{EPDSO}. Therefore $u=\varphi_{1}u+\varphi_{2}u\in H^{2}(\Omega)$, and this proves this case.

The case of the operator $H^{A}_{N}$ is carried by a similar argument, in fact, if $u\in \mathrm{Dom}\,H^{A}_{\mathrm N}$ and $\varphi_{i}$ $i=1,2$ are defined as above, by a similar argument it is easy to see that $\varphi_{1}u\in\mathrm{Dom}\,H^{A}_{\mathrm N,1}\subset H^{2}( B_{r}\cap\Omega)$ where $H^{A}_{\mathrm N,1}$ is the Neumann operator associated with $H^{A}$ in $\Omega\setminus B_{r/2}$ , and  $\varphi_{2}u\in \mathrm{Dom}\triangle_{\mathrm D,2}\subset H^{2}(\Omega\setminus B_{r/2})$, thus, $u=\varphi_{1}u+\varphi_{2}u\in H^{2}(\Omega)$ and this finishes the proof.
\end{proof}

\subsection{Regularized Neumann trace in quasi-convex domains}
In the following we recall the concept of regularized Neumann trace map~\cite{ADO} with suitable adaptions to the magnetic context; it will allow us to construct a boundary triple for the operator~$H^{A}_{\mathrm{max}}$.
 
 In Subsection~\ref{sectQCD} we saw that, for quasi-convex domains, $\mathrm{Dom}\,H^{A}_{\mathrm D}\subset H^{2}(\Omega)$; hence,  $\mathrm{Dom}\,H^{A}_{\mathrm D}=H^{2}(\Omega)\cap H^{1}_{0}(\Omega)$. We are ready to introduce the operator~$\tau_{\mathrm{N}}^{A}$, called the {\it regularized magnetic Neumann trace}, as follows.
 
 \begin{definition}
 Let~$\Omega$ be a quasi-convex domain and   $z\in \mathbb{C}\setminus \sigma(H^{A}_{\mathrm D})$. The operator~$\tau_{\mathrm{N}}^{A}$ is defined by
 \begin{eqnarray*}
 \tau_{\mathrm{N}}^{A,z}&:&\mathrm{Dom}\,H^{A}_{\mathrm{max}}\rightarrow N^{\frac{1}{2}}(\partial\Omega)\,,\\
 \tau_{\mathrm{N}}^{A,z}u&:=&\gamma_{\mathrm{N}}^{A}(H^{A}_{\mathrm D}-z)^{-1}(H^{A}_{\mathrm{max}}-z)u\,,\quad u\in \mathrm{Dom}\,H^{A}_{\mathrm{max}}.
 \end{eqnarray*}
 Clearly, this operator is well defined and bounded.
 \end{definition}

\begin{theorem}\label{teotau}
Let~$\Omega$ be a quasi-convex domain and   $z\in \mathbb{C}\setminus \sigma(H^{A}_{\mathrm D})$, then the operator  $\tau_{\mathrm{N}}^{A,z}$ has the following properties:

i) $\tau_{\mathrm{N}}^{A,z}(H^{2}(\Omega)\cap H^{1}_{0}(\Omega))=N^{\frac{1}{2}}(\partial\Omega)$.
 
ii) $\mathrm{Ker}\, \tau_{\mathrm{N}}^{A,z}= H^{2}_{0}(\Omega)\dotplus\left\{u\in \mathrm{L}^{2}(\Omega)|\; (H^{A}-z)u=0\right\}$
 
iii) $ (H^{A}_{\mathrm{max}}u,v)_{\mathrm{L}^{2}(\Omega)}-(u,H^{A}_{\mathrm{max}}v)_{\mathrm{L}^{2}(\Omega)}=-\langle \tau_{\mathrm{N}}^{A,z}u,{\hat\gamma_{\mathrm{D}}}v\rangle_{ N^{\frac{1}{2}}(\partial\Omega)}+\overline{\langle \tau_{\mathrm{N}}^{A,\overline{z}}v,{\hat\gamma_{\mathrm{D}}}u\rangle}_{ N^{\frac{1}{2}}(\partial\Omega)}$, for all $u,v\in \mathrm{Dom}\,H^{A}_{\mathrm{max}}$.
\end{theorem} 
\begin{proof} i)Note that, for $z\in \mathbb{C}\setminus \sigma(H^{A}_{\mathrm D})$, the operator  $(H^{A}-z):H^{2}(\Omega)\cap H^{1}_{0}(\Omega)\rightarrow \mathrm{L}^{2}(\Omega)$ is onto, so the result follows by Lemma~\ref{tnh2h10}.

ii) It is clear that $H^{2}_{0}(\Omega)\dotplus\left\{u\in \mathrm{L}^{2}(\Omega)|\; (H^{A}-z)u=0\right\}\subset\mathrm{Ker}\, \tau_{\mathrm{N}}^{A,z}$. The fact that we have a direct sum follows from the invertibility of $H^{A}_{\mathrm D}-z$. Now, let  $u\in \mathrm{Ker}\, \tau_{\mathrm{N}}^{A,z}$, and set $w=(H^{A}_{\mathrm D}-z)^{-1}(H^{A}-z)u$. Then, $w\in H^{2}(\Omega)\cap H^{1}_{0}(\Omega)$ and ${\tilde\gamma_{\mathrm{N}}^{A}}w=\tau_{\mathrm{N}}^{A,z}u=0$,  thus, $w\in H^{2}_{0}(\Omega)$,  and  $u=(u-w)+w$,  with $w\in H^{2}(\Omega)$ and $u-w\in \mathrm{L}^{2}(\Omega)$ with $(H^{A}-z)(u-w)=0$, so, $\mathrm{Ker}\, \tau_{\mathrm{N}}^{A,z}\subset H^{2}_{0}(\Omega)\dotplus\left\{u\in \mathrm{L}^{2}(\Omega)|\; (H^{A}-z)u=0\right\}$.

iii) Take $u,v \in \left\{u\in \mathrm{L}^{2}(\Omega)|\; H^{A}u\in \mathrm{L}^{2}(\Omega)\right\}$ and put
$\tilde{u}=(H^{A}_{\mathrm D}-z)^{-1}(H^{A}-z)u$ and $\tilde{v}=(H^{A}_{\mathrm D}-\overline{z})^{-1}(H^{A}-\overline{z})v$; both elements are in $ H^{2}(\Omega)\cap H^{1}_{0}(\Omega)$. We have $(H^{A}-z)\tilde{u}=(H^{A}-z)u$ and $(H^{A}-\overline{z})\tilde{v}=(H^{A}-\overline{z})v$, furthermore  ${\tilde\gamma_{\mathrm{N}}^{A}}\tilde{u}=\tau_{\mathrm{N}}^{A,z}u$ and ${\tilde\gamma_{\mathrm{N}}^{A}}\tilde{v}=\tau_{\mathrm{N}}^{A,\overline{z}}v$, thus,
\begin{eqnarray*}
(H^{A}u,v)_{\mathrm{L}^{2}(\Omega)}&-&(u,H^{A}v)_{\mathrm{L}^{2}(\Omega)}=((H^{A}-z)u,v)_{\mathrm{L}^{2}(\Omega)}-(u,(H^{A}-\overline{z})v)_{\mathrm{L}^{2}(\Omega)}\\
&=&((H^{A}-z)\tilde{u},v)_{\mathrm{L}^{2}(\Omega)}-(u,(H^{A}-\overline{z})\tilde{v})_{\mathrm{L}^{2}(\Omega)}\\
&=&(\tilde{u},(H^{A}-\overline{z})v)_{\mathrm{L}^{2}(\Omega)}-\langle {\tilde\gamma_{\mathrm{N}}^{A}}\tilde{u},{\hat\gamma_{\mathrm{D}}}v\rangle_{ N^{\frac{1}{2}}(\partial\Omega)}-(u,(H^{A}-\overline{z})\tilde{v})_{\mathrm{L}^{2}(\Omega)}\\
&=&(\tilde{u},(H^{A}-\overline{z})v)_{\mathrm{L}^{2}(\Omega)}-(u,(H^{A}-\overline{z})\tilde{v})_{\mathrm{L}^{2}(\Omega)}-\langle \tau_{\mathrm{N}}^{A,z}u,{\hat\gamma_{\mathrm{D}}}v\rangle_{ N^{\frac{1}{2}}(\partial\Omega)}\\
&=&(\tilde{u}-u,(H^{A}-\overline{z})v)_{\mathrm{L}^{2}(\Omega)}-\langle \tau_{\mathrm{N}}^{A,z}u,{\hat\gamma_{\mathrm{D}}}v\rangle_{ N^{\frac{1}{2}}(\partial\Omega)}\\
&=&\langle \overline{\tau_{\mathrm{N}}^{A,\overline{z}}v,{\hat\gamma_{\mathrm{D}}}u\rangle}_{ N^{\frac{1}{2}}(\partial\Omega)} -\langle \tau_{\mathrm{N}}^{A,z}u,{\hat\gamma_{\mathrm{D}}}v\rangle_{ N^{\frac{1}{2}}(\partial\Omega)}\,,
\end{eqnarray*}
where in the second equation we have employed~\eqref{intp2}, in the third we have used that ${\tilde\gamma_{\mathrm{N}}^{A}}\tilde{u}=\tau_{\mathrm{N}}^{A,z}u$ and, in the fifth~\eqref{intp2}, $(H^{A}-z)(\tilde{u}-u)=0$ and ${\hat\gamma_{\mathrm{D}}}(\tilde{u}-u)=-{\hat\gamma_{\mathrm{D}}}(u)$. 
\end{proof}

\begin{lemma}\label{lemma61}
Let~$\Omega$ be a quasi-convex domain. Then the operator  ${\hat\gamma_{\mathrm{D}}}$, which satisfies~\eqref{gammadmax}, is onto. More specifically, for any  fixed $z\in \mathbb C\setminus\sigma(H^{A}_{\mathrm D})$, there exists  a linear and bounded application $ (N^{\frac{1}{2}}(\partial\Omega))^{*}\ni \theta\rightarrow u_{\theta}\in \mathrm{L}^{2}(\Omega)$ such that ${\hat\gamma_{\mathrm{D}}}u_{\theta}=\theta$ and $(H^{A}-z)u_{\theta}=0$.

\end{lemma}
\begin{proof}
Fix $z\in \mathbb C\setminus\sigma(H^{A}_{\mathrm D})$. By the observation at the beginning of this section, we can write
\begin{equation*}
(H^{A}_{\mathrm D}-\overline{z})^{-1}:\mathrm{L}^{2}(\Omega)\rightarrow H^{2}(\Omega)\cap H^{1}_{0}(\Omega),
\end{equation*}
which is a bounded map. Thus, given $\theta \in  (N^{\frac{1}{2}}(\partial\Omega))^{*}$,  using Lemma~\ref{tnh2h10}, the antilinear functional
\begin{equation*}
\theta\gamma_{\mathrm{N}}^{A}(H^{A}_{\mathrm D}-\overline{z})^{-1}:\mathrm{L}^{2}(\Omega)\rightarrow \mathbb C 
\end{equation*}
is well defined and bounded.
Then, by Riesz Theorem, there exists a unique ${\tilde u_{\theta}}\in \mathrm{L}^{2}(\Omega)$ such that
\begin{equation}\label{*}
\theta\gamma_{\mathrm{N}}^{A}(H^{A}_{\mathrm D}-\overline{z})^{-1}(f)=(f,{\tilde u_{\theta}})_{\mathrm{L}^{2}(\Omega)}
\end{equation}
and $\|{\tilde u_{\theta}}\|_{\mathrm{L}^{2}(\Omega)}=\|\theta\gamma_{\mathrm{N}}^{A}(H^{A}_{\mathrm D}-\overline{z})^{-1}\|_{B(\mathrm{L}^{2}(\Omega,C))}\leq C \|\theta\|_{(N^{\frac{1}{2}}(\partial\Omega))^{*}}$. Furthermore, the  application $\theta\rightarrow {\tilde u_{\theta}}$ is linear. By setting $f=(H^{A}-\overline{z})w$, with $w\in H^{2}(\Omega)\cap H^{1}(\Omega)$, in~\eqref{*}, we have
\begin{equation*}
\theta(\gamma_{\mathrm{N}}^{A}(w))=((H^{A}-\overline{z})w,{\tilde u_{\theta}})_{\mathrm{L}^{2}(\Omega)},\quad \forall w \in  H^{2}(\Omega)\cap H^{1}(\Omega).
\end{equation*}
In particular, for $w\in \mathrm{C}^{\infty}_{0}(\Omega)$, we have $((H^{A}-\overline{z})w,{\tilde u_{\theta}})_{\mathrm{L}^{2}(\Omega)}=\theta(\gamma_{\mathrm{N}}^{A}(w))=\theta(0)=0$, so, $(H^{A}-z){\tilde u_{\theta}}=0$ in the sense of distributions. Therefore, for $w\in H^{2}(\Omega)\cap H^{1}_{0}(\Omega)$ we can write
\begin{eqnarray*}
\theta(\gamma_{\mathrm{N}}^{A}(w))&=&((H^{A}-\overline{z})w,{\tilde u_{\theta}})_{\mathrm{L}^{2}(\Omega)}-(w,(H^{A}-z){\tilde u_{\theta}})\\
&=&-\langle \gamma_{\mathrm{N}}^{A}w,{\hat\gamma_{\mathrm{D}}}{\tilde u_{\theta}}\rangle_{N^{\frac{1}{2}}(\partial\Omega)}
\end{eqnarray*}
where in  the second equation we have employed~\eqref{intp2}. Since $\gamma_{\mathrm{N}}^{A}:  H^{2}(\Omega)\cap H^{1}_{0}(\Omega)\rightarrow N^{\frac{1}{2}}(\partial\Omega)$ is onto, we conclude that $\hat{\gamma}_{\mathrm{D}}(-{\tilde u_{\theta}})= \theta$. Thus, $u_{\theta}=-{\tilde u_{\theta}}$ satisfies the properties of the lemma, and this finishes the proof.
\end{proof}

\begin{lemma}\label{solgamman1/2}
Let~$\Omega$ be a quasi-convex domain. Take $z\in \mathbb C\setminus \sigma(H^{A}_{\mathrm D})$, then, for all  $f\in \mathrm{L}^{2}(\Omega)$ and all $g\in (N^{\frac{1}{2}}(\partial\Omega))^{*}$, the following boundary value problem has a unique solution  $u_{D}$:
\begin{eqnarray*}
(H^{A}-z)u&=&f\quad\mathrm{in}\quad\Omega\,,
\\
u&\in& \mathrm{L}^{2}(\Omega)\,,
\\
{\hat\gamma_{\mathrm{D}}}u&=&g\quad\mathrm{in}\quad\partial\Omega\,.
\end{eqnarray*}
Furthermore, there exists $C(z)>0$ such that 
\[
\|u_{D}\|_{\mathrm{L}^{2}(\Omega)}\leq C(z)\big\{\|f\|_{\mathrm{L}^{2}(\Omega)}+\|g\|_{(N^{\frac{1}{2}}(\partial\Omega))^{*}}\big\}.
\] In particular, if $g=0$ then $u_{D}\in H^{2}(\Omega)\cap H^{1}_{0}(\Omega)$.
\end{lemma}
\begin{proof}
By  Lemma~\ref{lemma61}, we can take $v\in \mathrm{L}^{2}(\Omega)$ such that $H^{A}v=0$, ${\hat\gamma_{\mathrm{D}}}v=g$ and $\|v\|_{\mathrm{L}^{2}(\Omega)}\leq C\|g\|_{(N^{\frac{1}{2}}(\partial\Omega))^{*}}$. Define  $w=(H^{A}_{\mathrm D}-z)^{-1}(f+zv)\in H^{2}(\Omega)\cap H^{1}_{0}(\Omega)$, then
\begin{equation*}
\|w\|_{\mathrm{L}^{2}(\Omega)}\leq C \|f+zv\|_{\mathrm{L}^{2}(\Omega)}\leq C'\big\{\|f\|_{\mathrm{L}^{2}(\Omega)}+\|g\|_{(N^{\frac{1}{2}}(\partial\Omega))^{*}}\big\}.
\end{equation*}
It is easy to see that $u=v+w$ is a solution of the problem that satisfies the inequality in the statement of the lemma. To see that the solution is unique we need to show that, if $u\in \mathrm{L}^{2}(\Omega)$ satisfy $(H^{A}-z)u=0$ and ${\hat\gamma_{\mathrm{D}}}u=0$, then $u=0$. Fix $f\in \mathrm{L}^{2}(\Omega)$ and let  $u_{f}=(H^{A}_{D}(\Omega)-\overline{z})^{-1}f\in H^{2}(\Omega)\cap H^{1}_{0}(\Omega)$, then,
\begin{eqnarray*}
(f,u)_{\mathrm{L}^{2}(\Omega)}&=&((H^{A}_{D}(\Omega)-\overline{z})u_{f},u)_{\mathrm{L}^{2}(\Omega)}\\&=&(u_{f},(H^{A}_{D}(\Omega)-z)u)_{\mathrm{L}^{2}(\Omega)}-\langle\gamma_{\mathrm{N}}^{A}u_{f},{\hat\gamma_{\mathrm{D}}}u\rangle_{N^{\frac{1}{2}}(\partial\Omega)}=0,
\end{eqnarray*}

where the second equality follow from \ref{intp2}, thus, it follows that $u=0$.
\end{proof}

The following corollary is an immediate consequence of the above two  lemmas.
 
\begin{corollary}\label{gammaisom}
Let~$\Omega$ be quasi-convex and $z\in \mathbb C\setminus\sigma(H^{A}_{\mathrm D})$. Then the operator
\begin{equation*}
{\hat\gamma_{\mathrm{D}}}:\left\{ u\in \mathrm{L}^{2}(\Omega)|\; (H^{A}-z)u=0\right\}\rightarrow (N^{\frac{1}{2}}(\partial\Omega))^{*}
\end{equation*}
is a continuous isomorphism with a continuous inverse.
\end{corollary} 

Let~$\Omega$ be a quasi-convex domain; due to the above results, we can introduce the operator 
\[
M^{A}_{\mathrm{D,N}}(z): (N^{\frac{1}{2}}(\partial\Omega))^{*}\rightarrow (N^{\frac{3}{2}}(\partial\Omega))^{*}
\] for $z\in \mathbb{C}\setminus\sigma(H^{A}_{\mathrm D})$, defined by $M^{A}_{\mathrm{D,N}}(z)f:=-{\hat\gamma}^A_{\mathrm{N}}u_{D}$, where $u_{D}$ is the unique solution of
\begin{equation*}
(H^{A}-z)u=0,\quad u\in \mathrm{L}^{2}(\Omega),\quad {\hat\gamma_{\mathrm{D}}}u=f.
\end{equation*}
This  operator is clearly bounded and 
\begin{equation*}
\tau^{A,z}_{\mathrm{N}}u={\hat\gamma}^A_{\mathrm{N}}u+M^{A}_{\mathrm{D,N}}(z)({\hat\gamma_{\mathrm{D}}}u),\quad \forall u\in \mathrm{Dom}H^{A}_{\mathrm{max}}.
\end{equation*}

\section{Boundary triples, self-adjoint extensions of $H^{A}_{\mathrm{min}}$ and gauge transformations}\label{cap7}
In this section, first we  briefly review the concept of boundary triples~\cite{SOSE}. Then we construct a boundary triple for the operator~$H^{A}_{\mathrm{max}}$ in the case of quasi-convex~$\Omega$, and  so we  obtain a parametrization of all self-adjoint extensions of~$H^{A}_{\mathrm{min}}$ through unitary operators on $N^{\frac{1}{2}}(\partial\Omega)$.  For the case $A=0$, this result gives us another parametrization to the family of all self-adjoint extension of the minimal Laplacian that is different from the one obtained in \cite{ADO}, where the parametrization is given in terms of self-adjoint operators defined in closed subspaces of~$(N^{1/2}(\partial\Omega))^{*}$, and is obtained via Theorem~II.2.1 of~\cite{ACON}. We think that our parametrization of the family of all self-adjoint extensions of $H^{A}_{\mathrm{min}}$ is  more direct and easier to apply.  

In this section the space $N^{\frac{1}{2}}(\partial\Omega)$ is considered a Hilbert space equipped with the inner product $(\cdot,\cdot)_{N^{\frac{1}{2}}(\partial\Omega)}$ and the associated norm will be denoted by $\|\cdot\|_{N^{\frac{1}{2}}(\partial\Omega)}$. Furthermore, the inner product $(\cdot,\cdot)_{N^{\frac{1}{2}}(\partial\Omega)}$ will be  selected in such way that it has the following property: let  $F$ be a measurable function such that $|F(x)|=1$ a.s., then $M_{F}$, the operator of multiplication by  $F$ in $N^{\frac{1}{2}}(\partial\Omega)$, is unitary and  $(M_{F})^*=M_{F^{-1}}$; an  inner  product of this kind can  always be constructed  in $N^{\frac{1}{2}}(\partial\Omega)$, see, for example, Remark~6.14 in~\cite{ADO}, where the authors assume that $\Omega$ is bounded, but (by the same proof) the result  holds true for  unbounded domains with compact boundary. We note that  Theorem~\ref{extaaU} does not depend of this particular choice of  inner product in $N^{\frac{1}{2}}(\partial\Omega)$, however, Theorem~\ref{gauge} does use this particular kind of inner product (as well as all results that follow from this theorem).

Denote by $I$ the surjective isometry $I:(N^{\frac{1}{2}}(\partial\Omega))^{*}\rightarrow N^{\frac{1}{2}}(\partial\Omega)$ defined by
\begin{equation*}
\langle f,g\rangle_{N^{\frac{1}{2}}(\partial\Omega)}=(f,Ig)_{N^{\frac{1}{2}}(\partial\Omega)},\qquad \forall f\in N^{\frac{1}{2}}(\partial\Omega)\quad\textrm{and}\quad \forall g\in(N^{\frac{1}{2}}(\partial\Omega))^{*}.
\end{equation*}

\subsection{Boundary triples}

\begin{definition}\label{detrif}
Let~$A$ be a linear, densely defined and closed operator in a Hilbert space~$H$. A boundary triple for~$A$ is a triple $(G,\Gamma_{1},\Gamma_{2})$, where~$G$ is a Hilbert space and $\Gamma_{1},\Gamma_{2}:\mathrm{Dom}\,A \rightarrow G$ are linear operators such that:

i)~$
\langle f,Ag\rangle-\langle Af,g\rangle=\langle\Gamma_{1}f,\Gamma_{2}g\rangle-\langle\Gamma_{2}f,\Gamma_{1}g\rangle$ for all $f,g\in \mathrm{Dom}\,A.
$

ii) The operator $(\Gamma_{1},\Gamma_{2}):\mathrm{Dom}\,A\rightarrow G \times G$ is onto.

iii) The set $\mathrm{Ker}(\Gamma_{1},\Gamma_{2})$ is dense in~$H$.
\end{definition}

By Theorems~1.2 and~1.12 of~\cite{SOSE}, we have

\begin{theorem}\label{extbtrip}
Let~$A$ be a linear, symmetric densely defined and closed operator in a Hilbert space~$H$ and  $(G,\Gamma_{1},\Gamma_{2})$ a boundary triple for $A^{*}$. Then, the application  that maps to each unitary operator~$U$ of~$G$ the operator~$A_{U}$ defined by
\begin{eqnarray*}
\mathrm{Dom}\,A_{U}&=&\left\{u\in \mathrm{Dom}\,A^{*}|\; i(\textbf{1}+U)\Gamma_{1}u=(\textbf{1}-U)\Gamma_{2}u\right\},\\
A_{U}u&:=&A^{*}u\,,\quad u\in \mathrm{Dom}\,A_{U},
\end{eqnarray*}
sets a bijection  between the set of unitary applications of~$G$ and the set of all self-adjoint extensions of~$A$.
\end{theorem}

\subsection{Self-adjoint extensions of  $H^{A}_{\mathrm{min}}$ in quasi-convex domains}
The next theorem establishes that $(N^{\frac{1}{2}}(\partial\Omega),\tau_{\mathrm{N}}^{A,z},{\hat\Gamma_{D}})$, with $z\in\mathbb{R}\setminus\sigma(H^{A}_{\mathrm D})$ and ${\hat\Gamma_{D}}=I\circ {\hat\gamma_{\mathrm{D}}}$, is a boundary triple for the operator  $H^{A}_{\mathrm{max}}=(H^{A}_{0})^{*}$; this result will be a tool for our description of the family of all self-adjoint extensions of~$H^{A}_{\mathrm{min}}$.

\begin{theorem}\label{teorBTHmax}
Let~$\Omega$ be a quasi-convex domain with compact boundary,  a vector field ${A}$ with components in $\mathrm{W}^{1}_{\infty}(\overline\Omega)$ and $z\in\mathbb{R}\setminus\sigma(H^{A}_{\mathrm D})$, then $(N^{\frac{1}{2}}(\partial\Omega),\tau_{\mathrm{N}}^{A,z},{\hat\Gamma_{D}})$  is a boundary triple for~$H^{A}_{\mathrm{max}}$.
\end{theorem} 
\begin{proof}
By  item iii) of Theorem~\ref{teotau}, for all $u,v\in \mathrm{Dom}\,H^{A}_{\mathrm{max}}$ we have
\begin{equation*}
(u,H^{A}_{\mathrm{max}}v)_{\mathrm{L}^{2}(\Omega)}-(H^{A}_{\mathrm{max}}u,v)_{\mathrm{L}^{2}(\Omega)}=(\tau_{\mathrm{N}}^{A,z}u,{\hat\Gamma_{D}}v)_{ N^{\frac{1}{2}}(\partial\Omega)}-({\hat\Gamma_{D}}u,\tau_{\mathrm{N}}^{A,z}v)_{ N^{\frac{1}{2}}(\partial\Omega)},
\end{equation*}
and so item~i) of  Definition~\ref{detrif} is satisfied.

Item ii) of Definition~\ref{detrif} follows since $\tau_{\mathrm{N}}^{A,z}(\mathrm{Ker}\,{\hat\gamma_{\mathrm{D}}})= N^{\frac{1}{2}}(\partial\Omega)$ and ${\hat\Gamma_{D}}(\mathrm{Ker}\, \tau_{\mathrm{N}}^{A,z})=N^{\frac{1}{2}}(\partial\Omega)$, where the first equation is a consequence of  item~i) of Theorem~\ref{teotau}, whereas the second one is a consequence of  item~ii) of  Theorem~\ref{teotau}, together with the fact that, by  Lemma~\ref{solgamman1/2}, ${\hat\gamma_{\mathrm{D}}}(\left\{ u\in \mathrm{L}^{2}(\Omega)|\; (H^{A}_{\mathrm{max}}-z)u=0\right\})=(N^{\frac{1}{2}}(\partial\Omega))^{*}$.

Item iii) of Definition~\ref{detrif} is also satisfied since $\mathrm{C}^{\infty}_{0}(\Omega)$ is contained in $\mathrm{Ker}(\tau_{\mathrm{N}}^{A,z},{\hat\Gamma_{D}})$.
\end{proof}

Next a direct consequence of Theorems~\ref{extbtrip} and~\ref{teorBTHmax} combined with Lemma~\ref{lemmaBerndtLaplac}.

\begin{theorem}\label{extaaU}
Let~$\Omega$ be a quasi-convex domain with compact boundary, a vector field ${A}$ with components in $\mathrm{W}^{1}_{\infty}(\overline\Omega)$  and $z\in\mathbb{R}\setminus\sigma(H^{A}_{\mathrm D})$. Then the application that associates the operator $H^{A,z}_{U}$, defined by
\begin{eqnarray*}
\mathrm{Dom}\,H^{A,z}_{U}&=&\big\{u\in\mathrm{Dom}\,H^{A}_{\mathrm{max}}\mid\; i(\textbf{1}+U)\tau_{\mathrm{N}}^{A,z}u=(\textbf{1}-U){\hat\Gamma_{D}}u\big\}\,,\\
H^{A,z}_{U}u&:=&H^{A}u\,,
\end{eqnarray*}
to the unitary transformation~$U$ on $N^{\frac{1}{2}}(\partial\Omega)$ establishes a bijection between the set of such unitary transformations  and the set of all self-adjoint extensions of~$H^{A}_{\mathrm{min}}$.
\end{theorem} 

\begin{remark}\label{remarkBijecU}  Fix, for instance, $z=-1$, since $H^{A}_{D}$ is nonnegative, $-1\in\mathbb{R}\setminus\sigma(H^{A}_{\mathrm D})$, for all admissible~$A$. It is interesting to note that Theorem~\ref{extaaU} gives a parametrization of all self-adjoint extensions of $H^A_{\mathrm{min}}$ in terms of~$U$, which is independent of the magnetic potential~$A$. This establishes a natural bijection between the set of self-adjoint extensions $H^{A,-1}_{U}$ of~$H^A_{\mathrm{min}}$ and those $H^{B,-1}_{U}$ of~$H^B_{\mathrm{min}}$, for each pair of admissible magnetic potentials~$A$ and~$B$ through $H^{A,-1}_{U}\longleftrightarrow H^{B,-1}_{U}$.
\end{remark}

\begin{remark}\label{remABallExt}
A particular situation is for  unbounded connected quasi-convex $\Omega\subset\mathbb R^2$,with $\partial\Omega$ a simple closed curve, and $A$ is such that there is no magnetic field in~$\Omega$, i.e., $\mathrm {rot}\,  A=0$ there. This is a typical Aharonov-Bohm setting and, to the best of our knowledge, Theorem~\ref{extaaU} gives the first description of all self-adjoint extensions for this case and, furthermore, also for irregular solenoids~$\partial\Omega$. An important question in this context is to know when and which  self-adjoint extensions~$H^{A,z}_U$ are unitarily equivalent to some realization with zero magnetic potential, that is, the presence of~$A$ would be physically immaterial. We have something to say in the comments after Theorem~\ref{twii}.
\end{remark}
\begin{remark}\label{remarkN3/2}
A boundary triple for $H^{A}_{\mathrm{max}}$ can be constructed in such way that the space $N^{3/2}_{A}(\partial\Omega)$ plays the role of $N^{1/2}(\partial\Omega)$ and a regularization $\tau_{\mathrm{D}}^{A,z}$ of~${\hat\gamma_{\mathrm{D}}}$, analogous to the operator $\tau_{\mathrm{N}}^{A,z}$, defined by 
 \begin{eqnarray*}
 \tau_{\mathrm{D}}^{A,z}&:&\mathrm{Dom}\,H^{A}_{\mathrm{max}}\rightarrow N^{\frac{3}{2}}_{A}(\partial\Omega)\,,\\
 \tau_{\mathrm{D}}^{A,z}u&:=&\hat{\gamma}_{\mathrm{D}}(H^{A}_{\mathrm N}-z)^{-1}(H^{A}_{\mathrm{max}}-z)u\,,\quad u\in \mathrm{Dom}\,H^{A}_{\mathrm{max}}\quad \mathrm{and}\quad z\in\mathbb{C}\setminus \sigma(H^{A}_{N}), 
 \end{eqnarray*}
plays the hole of $\tau_{\mathrm{N}}^{A,z}$ in the construction of the boundary triple (from Theorem~\ref{teorBTHmax}). The construction is done in a very similar way to the previous one. The advantage of the boundary triple using the space $N^{1/2}(\partial\Omega)$ is that this space does not depend explicitly on the magnetic potential~$A$ (Remark~\ref{remarkBijecU} makes use of this fact). 
\end{remark}

\subsection{Gauge equivalence}
 In what follows we  introduce the concept of  gauge equivalence  of vector field in $ (\mathrm{W}^{1}_{\infty}(\overline\Omega))^{n}$, with connected domain~$\Omega$, and  discuss how the self-adjoint extensions of $H^{A}_{\mathrm{mim}}$, given by  Theorem~\ref{extaaU}, behave under such gauge transformations. If~$A=(A_{1},...,A_n)$  is a such vector field on~$\Omega\subset \mathbb R^n$, $n\ge2$, let~$\omega_{A}$ denote the differential 1-form associated with~$A$, that is, in Cartesian coordinates  $\omega_{A}=\sum_{i=1}^{n} A_{i}\,\dx_{i}$. In this subsection we suppose that~$\Omega$ is open and connected.

\begin{definition}\label{defGaugeEq}
Let~$\Omega$ be a connected Lipschitz domain and $A,B$ vector fields with components in $\mathrm{W}^{1}_{\infty}(\overline\Omega)$.  We say that~$A$ is (quantum) gauge equivalent to~$B$ if the following conditions hold:

i) $\mathrm{d}(\omega_{B}-\omega_{A})=0$.

ii) For each (smooth by parts) closed path  $\gamma$ in~$\Omega$,  there exists an integer $n_{\gamma}$ such that 
\[
\int_{\gamma} (\omega_{A}-\omega_{B})= 2\pi n_{\gamma}.
\]
\end{definition} 
 
Let  $A,B$ be two  gauge equivalent vector fields in~$\Omega$, fix $x_{0}\in\Omega$ and consider the function $F^{\Omega}=F^{\Omega}_{A,B}:\Omega\rightarrow \mathbb{C}$ given by  
\begin{equation}\label{eqFOmega}
F^{\Omega}(x):=e^{i\int_{\gamma_x} (\omega_{B}-\omega_{A})},
\end{equation} where $\gamma_x$ is a path in~$\Omega$ connecting~$x_{0}$ to~$x\in \Omega$; note that this function is well defined by  item~ii) in the above definition and  $|F^{\Omega}|=1$.  We have:

\begin{lemma}
Let  $A,B\in (\mathrm{W}^{1}_{\infty}(\overline\Omega)))^{n}$ be gauge equivalent vector fields in $\Omega$. Then, $\nabla F^{\Omega}=i(B-A) F^{\Omega}$; moreover $F^{\Omega} \in\mathrm{W}_\infty^{2}(\overline\Omega)\cap H^{1}(\Omega)$.
 \end{lemma}

\begin{proof}
It is enough  to prove the first statement, since the rest is an easy consequence of it. Fix $x_{0}\in \Omega$ and an open ball, $B_{x_{0}}\subset\Omega$, with center $x_{0}$  such that ${A}-{B}$ is $K$-Lipschitz in $B_{x_{0}}$, $K>0 $; the statement will be concluded if we show that it holds in $B_{x_{0}}$.  Note that for $x\in B_{x_{0}}$, we have  $F^{\Omega}(x)=F^{\Omega}(x_{0}) e^{i\phi(x)}$, where $\phi(x)=\int_{[x_{0},x]} (\omega_{A}-\omega_{B})$ and $[x_{0},x]$ is the line segment connecting $x_{0}$ to $x$; we have $\nabla \phi(x)=({B}-{A})(x)$ for $x\in B_{x_{0}}$, indeed, if $({B}-{A})(x)=(f_{1}(x),...,f_{n}(x))$
\begin{eqnarray*}
\partial^{j}\phi(x)&=&\partial^{j}\int_{0}^{1}[\sum^{n}_{i=1}f_{i}(x_{0}+t(x-x_{0}))(x^{i}-x_{0}^{i})]dt\\
&=&\int_{0}^{1}\partial^{j}[\sum^{n}_{i=1}f_{i}(x_{0}+t(x-x_{0}))(x^{i}-x_{0}^{i})]dt\\
&=&\int_{0}^{1}[f_{j}(x_{0}+t(x-x_{0}))+\sum^{n}_{i=1}\partial^{j}f_{i}(x_{0}+t(x-x_{0}))t(x^{i}-x_{0}^{i})]dt\\
&=&\int_{0}^{1}[f_{j}(x_{0}+t(x-x_{0}))+\sum^{n}_{i=1}\partial^{i}f_{j}(x_{0}+t(x-x_{0}))\,t(x^{i}-x_{0}^{i})]dt\\
&=&\int_{0}^{1}\frac{d}{dt}(tf_{j}(x_{0}+t(x-x_{0}))\\
&=&f_{j}(x),
\end{eqnarray*} 
where  the second equality is justified by an application of the dominated convergence theorem to the limit of the definition of the differential, which can be applied because the integrand is bounded since $({B}-{A})$ is $K$-Lipschitz in $B_{x_{0}}$; the fourth equality is a consequence of the fact that $ \partial^{i}f_{j}=\partial^{j}f_{i}$, for  $i,j=1,...,n$, and this proves the above statement.
\end{proof}
From this lemma the following function is also well defined, 
\begin{equation}\label{eqFU}
F^{\partial\Omega}:=\gamma_{\mathrm{D}}F^{\Omega},
\end{equation}
and for a unitary transformation~$U$ on $N^{\frac{1}{2}}(\partial\Omega)$, we define
\begin{equation}\label{eqFUpartial}
\mathcal F_U:=(F^{\partial\Omega})^{-1}U F^{\partial\Omega}.
\end{equation} 

\begin{theorem}\label{gauge}
Let~$\Omega$ be a quasi-convex domain with compact boundary. Let  $A,B$ be  gauge equivalent vector fields in $(\mathrm{W}^{1}_{\infty}(\overline\Omega))^{n}$. Then, with the same hypotheses and notations of Theorem~\ref{extaaU}, the self-adjoint extension  $H^{A,z}_{U}$ of~$H^{A}_{\mathrm{min}}$ is unitarily equivalent to  $H^{B,z}_{\mathcal F_U}$ (extension of~$H^{B}_{\mathrm{min}}$). More precisely,
\[
H^{A,z}_{U}(F^{\Omega}u)=F^{\Omega}H^{B,z}_{\mathcal F_U}u,
\]
for all $ u\in\mathrm{Dom}\,H^{B,z}_{\mathcal F_U}=(F^{\Omega})^{-1}\mathrm{Dom}\,H^{A,z}_{U}$ or, equivalently, $H^{A,z}_{U}F^{\Omega}= F^{\Omega}H^{B,z}_{\mathcal F_U}$.
\end{theorem}
\begin{proof}
 First note that $H^{B}_{D}=(F^{\Omega})^{-1}H^{A}_{\mathrm D}F^{\Omega}$ (by an  abuse of notation,  the multiplication operator by $F^{\Omega}$ is also denoted by $F^{\Omega}$). Indeed, in the proof of Proposition~\ref{opdirneu} we have verified that~$H^{A}_{\mathrm D}$ is the self-adjoint operator associated with the form $\Phi_{A,\mathrm{D}}$, the statement then follows  by $\Phi_{B,D}(u,v)=\Phi_{A,\mathrm{D}}(F^{\Omega}u,F^{\Omega}v)$, for all $u,v\in H^{1}_{0}(\Omega)$. 

 Note  now that $\mathrm{Dom}\,H^{A}_{\mathrm{max}}=F^{\Omega}\mathrm{Dom}\,H^{B}_{\mathrm{max}}$ and $H^{A}_{\mathrm{max}}(F^{\Omega}u)= F^{\Omega}H^{B}_{\mathrm{max}}u$. In fact, it is clear that if $u\in \mathrm{C}^{\infty}_{0}(\overline{\Omega})$, then $H^{A}_{\mathrm{max}}(F^{\Omega}u)= F^{\Omega}H^{B}_{\mathrm{max}}u$. Fix  $u\in \mathrm{Dom}\,H^{B}_{\mathrm{max}}$, and take a sequence $\{u_{j}\}_{j\in N}$ in $\mathrm{C}^{\infty}_{0}(\overline{\Omega})$   converging to~$u$ in $\mathrm{Dom}\,H^{B}_{\mathrm{max}}$, with the graph norm $\|\cdot\|_{B}=\|\cdot\|_{\mathrm{L}^{2}(\Omega)}+\|H^{B}(\cdot)\|_{\mathrm{L}^{2}(\Omega)}$, then, from what we said above, 
\begin{equation*}
 \|H^{A}(F^{\Omega}u_{j})-F^{\Omega}H^{B}u\|_{\mathrm{L}^{2}(\Omega)}=\|H^{B}(u_{j})-H^{B}(u)\|\rightarrow 0
\end{equation*}
 as $j\rightarrow \infty$; in particular, $\{F^{\Omega}u_{j}\}_{i\in N}$ is a Cauchy sequence in the graph norm of~$H^{A}_{\mathrm{max}}$ and, therefore, convergent in this space. On the other hand, since $(F^{\Omega}u_{i},H^{A}(F^{\Omega}u_{i}))$ converges in $\mathrm{L}^{2}(\Omega)\times \mathrm{L}^{2}(\Omega)$ to $(F^{\Omega}u,F^{\Omega}H^{B}(u))$, from the fact that~$H^{A}_{\mathrm{max}}$ is closed, it follows that $H^{A}(F^{\Omega}u)=F^{\Omega}H^{B}(u)\in \mathrm{L}^{2}(\Omega)$. Thus $F^{\Omega}\mathrm{Dom}\,H^{B}_{\mathrm{max}}\subset \mathrm{Dom}\,H^{A}_{\mathrm{max}}$ and $H^{A}(F^{\Omega}u)=F^{\Omega}H^{B}(u)\in \mathrm{L}^{2}(\Omega)$ for all $u\in \mathrm{Dom}\,H^{B}_{\mathrm{max}}$. Exchanging the roles of~$A$ and~$B$ in the above arguments, we see that the converse inclusion holds, and this proves the statement.

Note now that, for  $u\in \mathrm{C}^{\infty}_{0}(\overline{\Omega})$,
\begin{eqnarray*}
\tau_{\mathrm{N}}^{A,z}(F^{\Omega}u)&=&\gamma_{\mathrm{N}}^{A}(H^{A}_{\mathrm D}-z)^{-1}(H^{A}-z)(F^{\Omega}u)\\
&=&\gamma_{\mathrm{N}}^{A}(H^{A}_{\mathrm D}-z)^{-1}(F^{\Omega}(H^{B}(u)-z))\\
&=&\gamma_{\mathrm{N}}^{A}F^{\Omega}((H^{B}_{D}-z)^{-1}(H^{B}(u)-z))\\ &=&\nu\cdot\gamma_{\mathrm{D}}\nabla_{A}(F^{\Omega}((H^{B}_{D}-z)^{-1}(H^{B}(u)-z))\\
&=&\nu\cdot\gamma_{\mathrm{D}}F^{\Omega}\nabla_{B}((H^{B}_{D}-z)^{-1}(H^{B}(u)-z))\\ &=&\nu\cdot F^{\partial\Omega}\gamma_{\mathrm{D}}\nabla_{B}((H^{B}_{D}-z)^{-1}(H^{B}(u)-z))\\ &=&F^{\partial\Omega}\tau_{\mathrm{N}}^{B,z}(u),
\end{eqnarray*}
then, using the last observation and the fact that $\mathrm{C}^{\infty}_{0}(\overline{\Omega})$ is dense in $\mathrm{Dom}\,H^{B}_{\mathrm{max}}$, we conclude that $\tau_{\mathrm{N}}^{A,z}(F^{\Omega}u)=F^{\partial\Omega}\tau_{\mathrm{N}}^{B,z}(u)$ for all $u\in \mathrm{Dom}\,H^{B}_{\mathrm{max}}$. Note, also, that for all $u\in \textrm{Dom }H^{B}_{\mathrm{max}}$ one has ${\hat\Gamma_{D}}(F^{\Omega}u)=F^{\partial\Omega}{\hat\Gamma_{D}}u$. Indeed, it is easy to see that this holds for $u\in \mathrm{C}^{\infty}_{0}(\overline{\Omega})$, and the general case is a consequence of the fact that $\mathrm{C}^{\infty}_{0}(\overline{\Omega})$ is dense in $\textrm{Dom }H^{B}_{\mathrm{max}}$.

With  these facts, to conclude the proof it is enough to show that 
\[
\mathrm{Dom}\,H^{B,z}_{\mathcal F_U}=(F^{\Omega})^{-1}\mathrm{Dom}\,H^{A,z}_{U},
\] or, equivalently,  $F^{\Omega}\mathrm{Dom}\,H^{B,z}_{\mathcal F_U}=\mathrm{Dom}\,H^{A,z}_{U}$. Pick then $u\in \mathrm{Dom}\,H^{B,z}_{\mathcal F_U}$;  it follows that  $F^{\Omega}u\in \mathrm{Dom}\,H^{A}_{\mathrm{max}}$ and  $\tau_{\mathrm{N}}^{A,z}(F^{\Omega}u)=F^{\partial\Omega}\tau_{\mathrm{N}}^{B,z}(u)$, so that
\begin{eqnarray*}
i(\textbf{1}+U)\tau_{\mathrm{N}}^{A,z}(F^{\Omega})&=&i(\textbf{1}+U)F^{\partial\Omega}\tau_{\mathrm{N}}^{B,z}u\\
&=&F^{\partial\Omega}i\Big[(\textbf{1}+\mathcal F_U)\Big]\tau_{\mathrm{N}}^{B,z}u\\
&=&F^{\partial\Omega}(\textbf{1}-\mathcal F_U){\hat\Gamma_{D}}u\\
&=&(\textbf{1}-U)F^{\partial\Omega}{\hat\Gamma_{D}}u\\
&=&(\textbf{1}-U){\hat\Gamma_{D}}(F^{\Omega}u),
\end{eqnarray*}
 and $F^{\Omega}u\in \mathrm{Dom}\,H^{A}_{U}$; then,  $F^{\Omega}\mathrm{Dom}\,H^{B}_{\mathcal F_U}\subset\mathrm{Dom}\,H^{A}_{U}$. Exchanging  the roles of~$A$ and~$B$ in the arguments, we obtain the converse inclusion, and this finishes the proof.
\end{proof}

Next, a direct consequence of Theorem~\ref{gauge}.

\begin{corollary}
Let  $A,B$ be vector fields,  ${A},{B}\in (\mathrm{W}^{1}_{\infty}(\overline\Omega))^{n}$   and satisfying $B=A+\nabla\Lambda$ with $\Lambda\in \mathrm{W}^{2}_{\infty}(\overline\Omega)$. Then, under the same hypotheses of  Theorem~\ref{extaaU}, the self-adjoint extension $H^{A,z}_{U}$ of~$H^{A}_{\mathrm{min}}$ is unitarily equivalent to the extension $H^{B,z}_{(e^{-i\lambda}Ue^{i\lambda})}$ of~$H^{B}_{\mathrm{min}}$ , with $\lambda=\gamma_{\mathrm{D}}\Lambda$. In fact, 
\[
H^{A,z}_{U}(e^{i\Lambda}u)=e^{i\Lambda}H^{B,z}_{(e^{-i\lambda}Ue^{i\lambda})}u,
\]
for all $ u\in\mathrm{Dom}\,H^{B,z}_{(e^{-i\lambda}Ue^{i\lambda})}=e^{-i\Lambda}\mathrm{Dom}\,H^{A,z}_{U}$.
\end{corollary}

\begin{remark}\label{remarkGaugeClassic}
One says that two magnetic potentials~$A$ and~$B$ are {\it classical gauge equivalent} if for each $x\in\Omega$ there is a smooth function~$\Lambda_x$, defined in a neighborhood of~$x$, so that $B=A+\nabla\Lambda_x$ holds in that neighborhood. This implies that~$A$ and~$B$ generate the same magnetic field, but it does not guarantee that such~$A$ and~$B$ are gauge equivalent in the sense of Definition~\ref{defGaugeEq}. This distinction is at the heart of the Aharonov-Bohm effect.
\end{remark}

\begin{theorem}\label{Helffer}
Let~$\Omega$ be a bounded connected Lipschitz domain and a vector field~$A$ with components in $\mathrm{W}^{1}_{\infty}(\overline\Omega)$. Recall that  $H^{0}_{D}=-\triangle_{D}$  (the Dirichlet Laplacian), and let  $\lambda_{0}$ and $\lambda_{A,0}$ be the first (lowest) eigenvalues of $-\triangle_{D}$ and~$H^{A}_{\mathrm D}$, respectively. Then, the following statements are equivalent:

i) $\lambda_{0}=\lambda_{A,0}$.

ii) $A$ is  gauge equivalent to $0$.

iii) $H^{A}_{\mathrm D}$ is unitarily equivalent to $-\triangle_{D}$.
\end{theorem} 

The proof of Theorem~\ref{Helffer} is identical to the proof of Proposition~1.1 in~\cite{EFAB}, although  our statement  is  slightly different. The differences are: 1)~Here $\Omega$ is supposed to be a bounded Lipschitz domain, not necessarily with a smooth boundary. 2)~Since we suppose that $\Omega$ is bounded, we do not need to add a scalar potential to~$H^{A}$,  that  diverges at  infinity,  to ensure that the resulting Dirichlet extension has discrete spectrum. 3)~The potential $A$ is  in $(\mathrm{W}^{1}_{\infty}(\overline\Omega))^{n}$ and is not smooth, as (in principle) assumed in~\cite{EFAB}. The main ingredient in this proof is the following identity
\[
\Big\|\Big(\nabla-i{A}-\frac{\nabla u_{0}}{u_{0}}\Big)\varphi\Big\|_{\mathrm{L}^{2}(\Omega)}^{2}=((H^{A}-\lambda_{0})\varphi,\varphi)_{\mathrm{L}^{2}(\Omega)},\quad \forall \varphi\in \textrm{C}^{\infty}_{0}(\Omega),
\]
where $u_{0}(x)>0$, for all $x\in \Omega$, is  an eigenvector associated with the first eigenvalue $\lambda_{0}$ of $-\Delta_{D}$, which is not degenerate. This identity remains valid under the hypothesis of Theorem~\ref{Helffer}. By applying Theorems~\ref{gauge} and~\ref{Helffer}, we will conclude:

\begin{theorem}\label{twii}
Let~$\Omega$ be a bounded connected and quasi-convex domain,    ${A}\in(\mathrm{W}^{1}_{\infty}(\overline\Omega))^{n}$. Let $\lambda_{0}$ and $\lambda_{A,0}$ be the first eigenvalue of  $-\triangle_{D}$ and~$H^{A}_{\mathrm D}$, respectively. Fix $z\in\mathbb{R}\setminus\sigma(H^{A}_{\mathrm D})$. Then the following statements are equivalent:

i) $\lambda_{0}=\lambda_{A,0}$.

ii) $A$ is  gauge equivalent to $0$.

iii) Let $F^{\Omega}=F^{\Omega}_{A,0}$ be given by~\eqref{eqFOmega}. Then for all unitary applications~$U$ on $N^{\frac{1}{2}}(\partial\Omega)$ one has, by using~\eqref{eqFUpartial},
\begin{equation}\label{eqGaugeEffect}
(F^{\Omega})^{-1}H^{A,z}_{U}F^{\Omega}\;= \;H^{0,z}_{\mathcal F_U}\,.
\end{equation}

iv) There exist unitary applications~$U$ and~$V$ on $N^{\frac{1}{2}}(\partial\Omega)$ such that 
\[
\mathcal J^{-1}H^{A,z}_{U}\mathcal J\;=\;H^{0,z}_{V},
\]  for some~$\mathcal J:\Omega\rightarrow \mathbb{C}$,  $\mathcal J \in\mathrm{W}^{2}_{\infty}(\overline\Omega)$, and with $|\mathcal{J}(x)|=1$ for all~$x\in\Omega$.
\end{theorem} 
\begin{proof}
The equivalence of~i) and~ii) is a consequence of Theorem~\ref{Helffer}. By Theorem~\ref{gauge}, ii) implies~iii), and the fact that iii) implies~iv) is obvious.
To see that~iv) implies~ii), note that  by Remark~\ref{rem3}, the operator of multiplication by $J$ maps $H^{2}_{0}(\Omega)$ into itself, since  $J\in \textrm{W}^{2}_{\infty}(\overline\Omega)$. By item~iv), $\mathcal{J}^{-1}H^{A,z}_{U}\mathcal{J}=H^{0,z}_{V}$ and we have
\begin{equation}\label{vimpii}
H^{0}_{\mathrm{mim}}u=H^{0,z}_{V}u=\mathcal{J}^{-1}H^{A,z}_{U}\mathcal{J}u=\mathcal{J}^{-1}H^{A,z}_{\mathrm{mim}}\mathcal{J}u,\qquad \forall u\in H^{2}_{0}(\Omega);
\end{equation}
therefore, $\mathcal{J}^{-1}H^{A}_{\mathrm{mim}}\mathcal{J}=H^{0}_{\mathrm{mim}}$, since, by  Lemma~\ref{lemmaBerndtLaplac}, $\mathrm{Dom}\, H^{A}_{\mathrm{mim}}=\mathrm{Dom}\, H^{0}_{\mathrm{mim}}= H^{2}_{0}(\Omega)$.  Denote $G=\mathcal{J}^{-1}$; so $G:\Omega\rightarrow S^{1}\subset \mathbb{R}^{2}$, and $G$ maps $H^{2}_{0}(\Omega)$ into itself  and  relation~\eqref{vimpii} is equivalent to
\[
-\triangle(Gu)=GH^{A}u,\quad \forall u\in H^{2}_{0}(\Omega);
\]
 in particular,
\[
-(G\triangle u+u\triangle G+2 \nabla G\cdot\nabla u)=G\{-\triangle u-2 i A \cdot \nabla u+(|A|^{2}-i\,\mathrm{div}A)u\},\quad \forall u\in \textrm{C}^{\infty}_{0}(\Omega),
\]
thus, 
\[
u\{\triangle G-(|A|^{2}-i\,\mathrm{div}A)G\}=\nabla u \cdot(-2\nabla G-2iG A),\quad \forall u\in \mathrm{C}^{\infty}_{0}(\Omega),
\]
 in particular, it follows that $\frac{\nabla G}{iG}=-A$, or, equivalently, $\omega_{A}=-G^{*}(\frac{dz}{iz})$ (here~$G^*$ is the pullback operation). In particular, $\omega_{A}$ is closed, and for each closed path $\gamma$ in~$\Omega$, we have
\[
\int_{\gamma}\omega_{A}=\int_{\gamma}-G^{*}\big(\frac{dz}{iz}\big)=\int_{G\gamma}-\frac{dz}{iz}=2 \pi n
\]
for some $n\in \mathbb{Z}$, and~ii) follows.
\end{proof}

By Theorem~\ref{twii}~iv), if just one extension of~$H^A_{\mathrm{min}}$ is unitarily equivalent (through the multiplication by a function~$\mathcal J$) to a realization with zero magnetic potential, then  the same occurs for all self-adjoint realizations of~$H^A_{\mathrm{min}}$, that is, the unitarily equivalences are always implemented by gauge transformations. And this is independent of the spectral type of each realization. Although these remarks are physically expected, it seems there was no mathematical proof in the literature yet; and here for (bounded) quasi-convex domains.  

For the Aharonov-Bohm effect we need a multiply-connected domain~$\Omega\subset \mathbb R^2$  and $\mathrm {rot}\,  A=0$ in~$\Omega$ (for simplicity we restrict the discussion to the plane);  for instance, a quasi-convex annulus (the inner and outer border given by simple closed curves; the inner border represents the solenoid and the outer one could be a circle); then Theorem~\ref{twii} says that if~$A$ gives no contribution in case of some boundary condition, or simply if the first Dirichlet eigenvalue coincides with that of the Laplacian (item~i) in the theorem), then~$A$ gives no contribution for all self-adjoint extensions of~$H^A_{\mathrm{min}}$.

\begin{remark}
Assume that~$A$ is such that the conclusions of Theorem~\ref{twii} apply, and consider the extensions parametrized by the unitary operators $U_1=-\textbf 1$ (that is the Dirichlet extension, in fact) and $U_2=\textbf 1$  on $N^{\frac{1}{2}}(\partial\Omega)$. In such cases, by~\eqref{eqGaugeEffect},  $H^{0,z}_{\mathcal F_{U_1}}=H^{0,z}_{{U_1}}$ and $H^{0,z}_{\mathcal F_{U_2}}=H^{0,z}_{U_2}$ also carries the same respective boundary conditions. But this direct correspondence (i.e., the same boundary conditions for unitarily equivalent operators in items~iii) and~iv) of the theorem) is not expected for all extensions; it is enough to pick a~$U\ne \mathcal F_U$.
\end{remark}

\end{document}